\documentclass[11pt]{article}

\usepackage{fullpage}

\usepackage[hidelinks]{hyperref}       
\usepackage{url}            
\usepackage{booktabs}       
\usepackage{amsfonts}       
\usepackage{nicefrac}       
\usepackage{microtype}      
\usepackage{amssymb}
\usepackage{amsmath}
\usepackage{amsthm}
\usepackage[numbers]{natbib}
\DeclareMathOperator{\E}{\mathbb{E}}
\usepackage{wrapfig}
\usepackage{color}
\usepackage{multirow}
\usepackage{rotating}
\usepackage{bbm}
\usepackage{enumitem}
\usepackage{subcaption}
\usepackage[noend,ruled,noline,linesnumbered]{algorithm2e}

\SetCommentSty{mycommfont}
\DontPrintSemicolon
\usepackage{eqparbox,array}
\usepackage{setspace}
\SetArgSty{textup}
\usepackage{tikz}
\usepackage{pgfplots}
\pgfplotsset{compat = newest}
\usepackage{float}
\usepackage{soul}
\newcommand\numberthis{\addtocounter{equation}{1}\tag{\theequation}}

\newtheorem{theorem}{Theorem}

\newtheorem{observation}[theorem]{Observation}
\newtheorem{definition}[theorem]{Definition}
\newtheorem{lemma}[theorem]{Lemma}

\usepackage{thm-restate}

\usepackage{authblk}
\theoremstyle{plain}
\newtheorem*{theorem*}{Theorem}
\newtheorem*{corollary*}{Corollary}

\renewcommand{\algorithmcfname}{MECHANISM}

\usepackage{xcolor}
\usepackage{color-edits}
\addauthor{eb}{blue}
\addauthor{cz}{red}
\addauthor{vg}{blue}
\addauthor{xt}{cyan}
\title{Online Mechanism Design with Predictions\thanks{All four authors were supported by NSF grant CCF-2210502. Eric Balkanski was also supported by NSF grant  IIS-2147361.}}
\author[a]{Eric Balkanski\thanks{eb3224@columbia.edu}}
\author[b]{Vasilis Gkatzelis\thanks{gkatz@drexel.edu}}
\author[b]{Xizhi Tan\thanks{xizhi@drexel.edu}}
\author[a]{Cherlin Zhu\thanks{cz2740@columbia.edu}}
\affil[a]{Columbia University, IEOR}
\affil[b]{Drexel University, Computer Science}
\date{}

\begin{document}

\maketitle

\begin{abstract}

Aiming to overcome some of the limitations of worst-case analysis, the recently proposed framework of ``algorithms with predictions'' allows algorithms to be augmented with a (possibly erroneous) machine-learned prediction that they can use as a guide. In this framework, the goal is to obtain improved guarantees when the prediction is correct, which is called \emph{consistency}, while simultaneously guaranteeing some worst-case bounds even when the prediction is arbitrarily wrong, which is called \emph{robustness}. The vast majority of the work on this framework has focused on a refined analysis of online algorithms augmented with predictions regarding the future input. A subsequent line of work has also successfully adapted this framework to mechanism design, where the prediction is regarding the private information of strategic agents. In this paper, we initiate the study of online mechanism design with predictions, which combines the challenges of online algorithms with predictions and mechanism design with predictions.

We consider the well-studied problem of designing a revenue-maximizing auction to sell a single item to strategic bidders who arrive and depart over time, each with an unknown, private, value for the item. We study the learning-augmented version of this problem where the auction designer is given a prediction regarding the maximum value over all agents. Our main result is a strategyproof mechanism whose revenue guarantees are $\alpha$-consistent with respect to the highest value and $(1-\alpha^2)/4$-robust with respect to the second-highest value, for $\alpha \in [0,1]$. We show that this trade-off is optimal within a broad and natural family of auctions, meaning that any $\alpha$-consistent mechanism in that family has robustness at most $(1-\alpha^2)/4$.
Finally, we extend our mechanism to also obtain expected revenue that is proportional to the prediction quality.
\end{abstract}

\thispagestyle{empty} 

\newcommand{\prediction}{\Tilde{v}_{(1)}}
\newcommand{\bestsofar}[1]{v^{\leq #1}_{\max}}
\newcommand{\rev}{\texttt{Rev}}
\newcommand{\vmax}{v_{(1)}}
\newcommand{\vbest}{v_{\max}}
\newcommand{\first}{i_1}
\newcommand{\second}{i_2}
\newcommand{\mech}{\textsc{Three-Phase}}
\newcommand{\emech}{\textsc{Error-Tolerant}}
\newcommand{\pref}{\succ}
\newcommand{\checkp}{\text{active-winner}}
\newcommand{\vmaxless}[1]{v_{\max}^{< #1}}
\newcommand{\vmaxlessbar}[1]{\overline{v}_{\max}^{< #1}}
\newcommand{\range}{W_n}
\newcommand{\quality}{q}
\newcommand{\suc}{\text{success}}
\newcommand{\tbest}{t_{(1)}}
\newcommand{\reach}{\text{reach i}}
\newcommand{\sofar}{v_i\text{ is the highest so far}}
\newcommand{\ji}{v_i^{max}<p}
\newcommand{\PP}{\mathbb{P}}
\newcommand{\vsf}{\mathcal{M}_{\text{vsf}}}
\newcommand{\stoprule}{\mathbf{s}}
\newcommand{\stoptime}{\sigma}
\newcommand{\infset}{\mathcal{I}_n(Y)}
\newcommand{\priceinterval}{I}
\newcommand{\randorder}{\mu}

\section{Introduction}
\setcounter{page}{1} 
One of the well-established shortcomings of worst-case analysis is that it often leads to overly pessimistic conclusions. On the other hand, any non-trivial performance guarantee that can be established through worst-case analysis is very robust since it holds no matter what the input may be. In an attempt to overcome the limitations of worst-case analysis without compromising its robustness, the recently proposed framework of ``algorithms with predictions'' allows algorithms to be augmented with a machine-learned prediction that they can use as a guide~\citep{mitzenmacher_vassilvitskii_2021}. Crucially, this prediction may be highly inaccurate, so depending too heavily on it can lead to very poor performance in the worst case. Therefore, the goal in this framework is to use such a prediction so that a strong performance can be guaranteed whenever the prediction is accurate (known as the \emph{consistency} guarantee) while simultaneously maintaining non-trivial worst-case guarantees even if the prediction is inaccurate (known as the \emph{robustness} guarantee).

Since this framework was introduced, a surge of work has utilized it toward a refined analysis of algorithms, data structures, and mechanisms (see \citep{alps}  for a frequently updated list of papers in this rapidly growing literature). The vast majority of this work has focused on the design and analysis of online algorithms, i.e., algorithms that need to process their input piece-by-piece and make irrevocable decisions without knowing the whole input. Learning-augmented online algorithms are enhanced with a prediction regarding the future input, which they can potentially use to make more informed decisions, while carefully managing the risk of being misguided by it. An even more recent line of work has successfully adapted this framework for the design and analysis of mechanisms interacting with strategic bidders \citep{ABGOT22,XL22}. One of the canonical problems in mechanism design is the design of auctions for selling goods to a group of strategic bidders, aiming to maximize the revenue. The main obstacle in achieving this goal is the fact that the amount that each bidder is willing to pay is private information that the designer needs to carefully elicit. Learning-augmented mechanisms are therefore enhanced with predictions regarding the value of this private information, which can potentially alleviate this obstacle.

In this work, we initiate the study of online mechanism design with predictions, bringing together the two lines of work on online algorithms with predictions and mechanism design with predictions. Specifically, we consider the problem of selling goods to strategic bidders that arrive and depart over time. 
This problem combines the challenges of both lines of work since the designer needs to carefully elicit the unknown, private, value of each bidder, while also not knowing (and being unable to elicit) the values of the bidders who have not yet arrived. In fact, designing an auction for such dynamic settings can be more demanding because, apart from the combined information limitations that the designer faces, the bidders may not only strategically misreport their value for the good(s) being sold, but also strategically misrepresent their arrival and departure times. 

The study of online mechanism design (without predictions) has previously received a lot of attention, given the many important applications that involve dynamic settings with bidders that arrive and depart over time~\citep{parkes_2007}. For example, the sale of airplane and theater seats or the sale of cars usually takes place over a period of time, during which interested buyers join the market and depart from it. As this happens, the seller may gradually adjust the prices of the goods being sold, aiming to maximize the revenue. These adjustments can be a function of the demand that the seller observes over time, but it is quite natural to assume that the designer may also have access to some prediction regarding this demand, e.g., using historical data. Our goal in this paper is to design online auctions enhanced with such a prediction and to evaluate the extent to which they can yield strong performance guarantees in terms of consistency and robustness.

\subsection{Our Results}
Our main goal is to evaluate the potential impact of the learning-augmented model on the performance of auctions in dynamic environments. To achieve this goal, we revisit the online mechanism design model, where the bidders arrive and depart over time~\citep{parkes_2007,HKP04}. This model poses several realistic and non-trivial obstacles for the auction designer: 1) the bidders can lie about their value for the good(s) being sold (the standard obstacle in mechanism design), 2) during the execution of the auction, the auctioneer has no information regarding bidders who have not yet arrived (the standard obstacle in online problems), and 3) the bidders can also lie regarding their arrival and departure times (an obstacle that is specific to the online mechanism design setting).

Within this model, we focus on the problem of selling a single item, aiming to maximize revenue. Each bidder $i$ has a value $v_i$ for the item being sold, and this value is the largest amount she would be willing to pay for it. In the absence of any predictions, the best revenue that one can guarantee, even in an offline setting, is equal to the second-highest value over all bidders.\footnote{This can be achieved by the classic Vickrey (second-price) auction.} Using this as a benchmark, prior work proposed an online single-item auction that guarantees revenue at least $1/4$ of the second-highest value~\citep{HKP04}. Aiming to refine this result and achieve stronger guarantees, we adopt the learning-augmented framework and consider the design of online auctions that are enhanced with a (possibly very inaccurate) prediction regarding the highest value over all bidders. The goal is to guarantee more revenue whenever the prediction is accurate (the consistency guarantee), while also achieving some non-trivial revenue guarantee even if the prediction is highly inaccurate (the robustness guarantee, which is equivalent to the worst-case guarantee studied in prior work). 

Targeting a more ambitious benchmark, we use the highest value over all bidders (the first-best revenue) as a benchmark for our consistency guarantee while maintaining the second-highest value (the second-best revenue) as the benchmark for robustness (as in prior work).

\paragraph{The $\mech$ learning-augmented online auction.}
Our first main result is the \mech\ auction: a learning-augmented online auction parameterized by some value $\alpha\in [0,1]$, which takes place in three phases. During the first phase, the auction observes the values of the first $\lceil\frac{1-\alpha}{2}n\rceil$ departing bidders in order to ``learn'' an estimate regarding what an appropriate price may be. In the second phase, the auction ``tests the prediction'' by giving each active bidder the opportunity to clinch the item if their value is at least as high as the prediction. After $\lfloor\alpha n\rfloor$ more bidders have departed, if the item remains unsold, the auction enters the third and last phase. During this phase, any active bidder is given the opportunity to clinch the item at a price equal to the maximum value observed during the first two phases. This learning-augmented online auction achieves the following trade-off between consistency and robustness. 
\begin{theorem*}
The \mech\ learning-augmented online auction is deterministic, strategyproof, and for any $\alpha \in [0,1]$ such that $\alpha n \in \mathbb{N} \text{ and } \frac{1-\alpha}{2}n \in \mathbb{N}$ the expected revenue it extracts guarantees $\alpha$-consistency with respect to the first-best revenue benchmark and $(1-\alpha^2)/4$-robustness with respect to the second-best revenue benchmark.
\end{theorem*}

Note that, although we focus on revenue maximization throughout this paper, as a corollary of our analysis, we also obtain a social welfare guarantee that is also $\alpha$-consistent and $(1-\alpha^2)/4$-robust, where consistency and robustness are both with respect to the highest value. 

\begin{figure}
    \centering
    \includegraphics[width = 0.55 \textwidth]{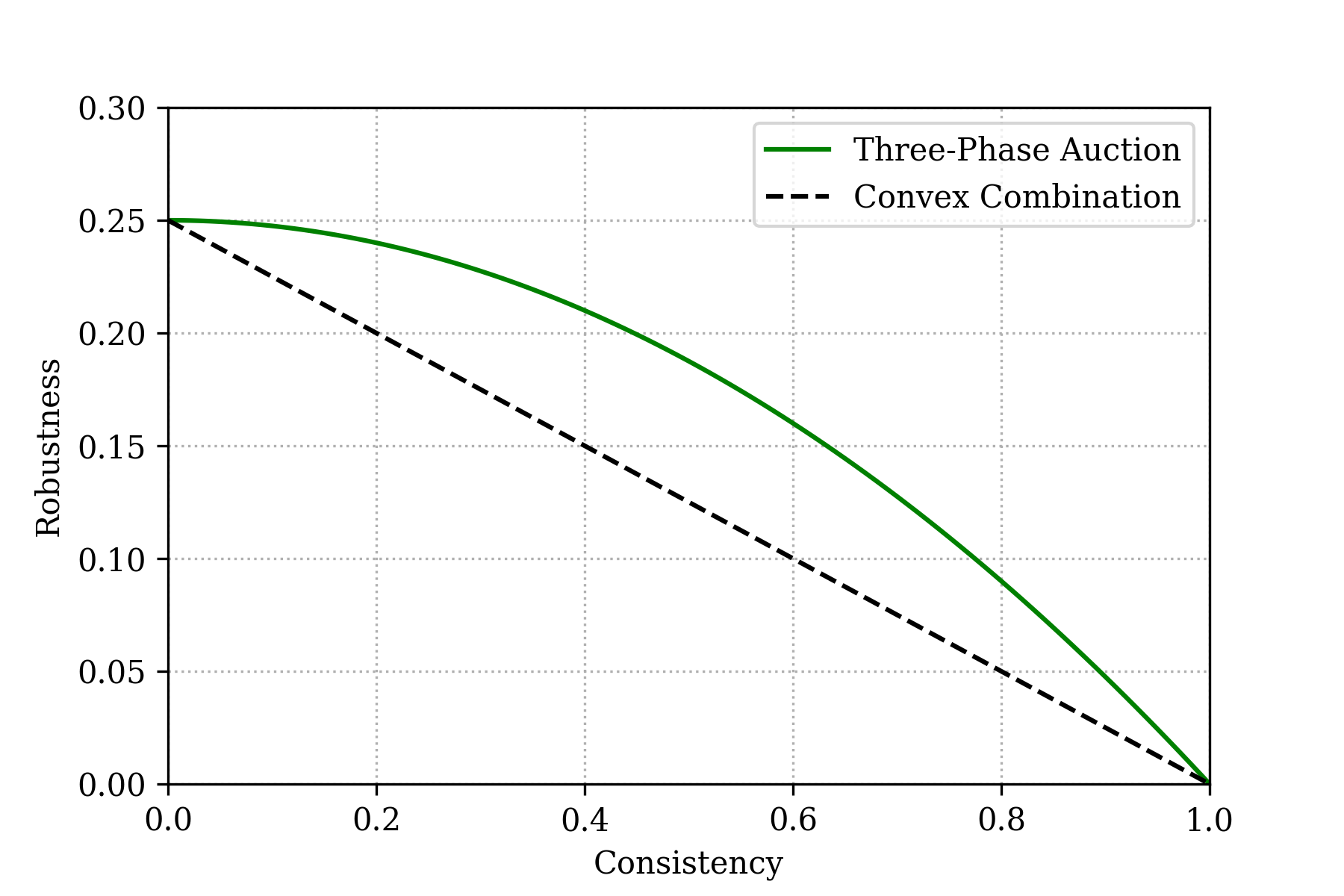}
    \caption{\small{The robustness-consistency trade-off achieved by the \mech\ auction and the trade-off achieved by  convex combinations of the auction that optimizes consistency by completely trusting the predictions and the auction that optimizes robustness by ignoring the predictions.}}
    \label{fig:errordistortion}
\end{figure}

If we let $\alpha=1$, our auction reduces to a posted price auction that offers to every bidder a price equal to the predicted maximum value, so it achieves a perfect consistency of $1$ (if the prediction is correct, then only a bidder with the maximum value would accept, and the optimal first-best revenue is extracted). In fact, this is the only way to achieve a consistency of $1$, but it provides no robustness guarantees (e.g., if the prediction was higher than the maximum value, then the good will remain unsold, leading to zero revenue). If the value of $\alpha$ is reduced, then our auction provides improved robustness guarantees in exchange for a drop in consistency, up until $\alpha=0$, for which our auction reduces to the one by~\cite{HKP04} and we retrieve the best-known robustness guarantee of $1/4$. The most interesting region for our mechanism is when $\alpha$ is in $(0,1)$ and all three of its phases are used in order to optimize the robustness-consistency trade-off while overcoming the strategyproofness obstacles that arise in the presence of multiple phases. Figure~\ref{fig:errordistortion} exhibits the convex combination of the two extreme solutions (optimizing only for consistency or only for robustness), as well as the improved trade-off achieved by the \mech\ auction.


\paragraph{Impossibility results.}
An interesting fact about the online mechanism design problem that we study is that proving impossible results for revenue maximization is significantly more demanding than proving impossibility results for the closely related secretary problem. 
First, the performance of the auction depends not only on who gets the good, but also at what price. For example, the auction could potentially offer the good to the highest value bidder for a price greater than the second-highest value, leading to revenue greater than the benchmark. Second, the price that the auction offers to a bidder can be an arbitrary function of the values observed among previous bidders, instead of just their relative ordering, giving rise to a very rich design space and making impossibility arguments quite demanding.  As a result, even though we would conjecture that $1/4$ is the optimal competitive ratio for revenue, the best-known inapproximability result remains at $2/3$.

We complement our positive results with impossibility results both for the original setting (without predictions) and the learning-augmented setting. For the original setting, we prove that the competitive ratio of $1/4$ is, in fact, optimal for a large family of auctions whose pricing rule can be an arbitrary function of the values observed in the past, but no more than the maximum seen so far. The proof relies on an LP formulation of the revenue maximization problem and  duality. For the learning-augmented setting, we prove that the robustness-consistency trade-off achieved by our \mech\ auction is optimal within a class of learning-augmented online auctions. We note that the previous LP approach requires a history-independence property that no longer holds in the setting with predictions. We therefore use an interchange argument instead to reduce an optimal auction from this class to our auction while maintaining the optimal revenue guarantee.

\subsection{Related Work}

\paragraph{Online mechanism design.} Due to the many applications with strategic agents who arrive and depart in an online fashion,  online mechanism design is an important subfield of mechanism design (see Chapter~16 by~\citet{parkes_2007} of the Algorithmic Game Theory textbook \citep{NRTV07} for an overview). The problem of designing strategyproof online auctions when bidders 
can misreport their arrival and  departure times was introduced by \citet{HKP04}. For revenue maximization, they show  a $1/4$-competitive strategyproof mechanism  and a $2/3$ impossibility result  in the single item setting. For the $k$-item setting, they give a constant factor competitive strategyproof mechanism, which was improved to  $1/(26e)$ for the special case where the active times of bidders do not overlap \citep{KT12}. For welfare maximization, \citet{HKP04} give a $1/e$-competitive strategyproof mechanism and a  $1/2$ impossibility in the single item setting. A  result by~\citet{correa} implies a tight $1/e$ impossibility for welfare. The impossibility results for welfare maximization hold even for mechanisms that are not strategyproof, whereas the strategyproofness assumption is required for the impossibility with respect to revenue maximization (otherwise, the mechanism could always charge the winning bidder their full value). 

We note that online mechanisms were first introduced in a setting where  the agents are unable to misreport their arrival and departure time \citep{lavi2000competitive}.
\citet{BJS10} consider agents with private arrival times and values, but without departure times, and propose a strategyproof auction that achieves a competitive ratio of $3/16$ in terms of revenue (the setting with arrival  and without departure times introduces challenges that do not occur with both arrival and departure times).  Additional related models for online mechanism design include unlimited supply and digital goods \citep{BYHW02, BRW03, BH05, KP13},
two-sided auctions with both buyers and sellers online \citep{BP12, BSZ06},
and interdependent value environments \citep{CIP07}.

\paragraph{Online algorithms with predictions.} The line of work on algorithms with predictions, also called learning-augmented algorithms, is an exciting emerging literature (see \citep{MV22} for a survey of early contributions and \citep{alps} for a frequently updated list of papers in this field). 
Numerous classic online algorithm design problems have been revisited, including online paging \citep{lykouris2018competitive}, scheduling \citep{PSK18}, optimization problems involving covering \citep{BMS20} and knapsack constraints \citep{IKQP21}, as well as Nash social welfare maximization \citep{banerjee2020online}, the secretary problem \citep{AGKK23, DLLV21, KY23}, and a range of graph-related challenges \citep{azar2022online}.

Among these previous works, the most closely related to our setting is by \citet{AGKK23}, who consider the value-maximizing secretary problem augmented with a prediction regarding the maximum value of the agents. In fact, their proposed learning-augmented algorithm follows a three-phase structure, which is similar to the one used in our auction. However, our setting and our proposed solution differ in several significant ways. The most significant one is the fact that just focusing on the online aspect of the problem, we need to deal with the important additional obstacle that the agents are strategic and can misreport their value, as well as their arrival and departure time. 
Furthermore, our goal is to maximize revenue whereas the goal in the secretary problem is to choose the agent with maximum value. Finally, our setting does not assume that each agent departs before the arrival of the next one, like the secretary setting does; this makes the problem of designing strategyproof mechanisms significantly more delicate.
To  achieve strategyproofness, our auction must very carefully determine the time at which the item is allocated, the bidder who receives the item, and the item’s price in order to handle bidders who might be active during multiple phases, which is the main technical (and novel) challenge in our setting.

\paragraph{Mechanism design with predictions} The design of learning-augmented mechanisms interacting with strategic agents is an even more recent line of work that was initiated by \citet{ABGOT22} and  \citet{XL22}. It includes strategic facility location~\citep{ABGOT22, XL22,IB22}, price of anarchy of cost-sharing protocols \citep{GKST22}, strategic scheduling~\citep{XL22, BGT223}, auctions~\citep{MV17,XL22,LuWanZhang23,caragiannis2024randomized}, and bicriteria (social welfare and revenue) mechanism design~\citep{BPS23}. We refer to \citep{BGT23} for a reading list of this line of work. Concurrently and independently of our work,  \citet{LuWanZhang23} studied online auctions, as well as multiple offline auctions settings, with predictions. However, the main technical and conceptual contribution of this work is in the offline setting; the online results are achieved through a reduction to the offline setting analogous to the one in~\cite{KP13}. Our results are, therefore, incomparable both from a technical and a conceptual standpoint.

\section{Preliminaries}
\label{sec:prelim}
We consider the problem of designing an auction to sell a single item to a set $N = \{1,2, \dots, n\}$ of $n$ bidders who arrive and depart over time. Each bidder $i$ arrives at some time $a_i$, departs at some time $d_i\geq a_i$, and has value $v_i$ for the item being sold. We refer to the interval $[a_i, d_i]$ as the \emph{active time} for bidder $i$. For simplicity, we assume that the bidders are indexed based on their order of departure (i.e., bidder $i$ is the $i$-th bidder to depart). We also let $\pi$ be an arbitrary total order over the set of bidders, which we use for tie-breaking, and let $i \pref j$ denote the fact that $i$ is ranked before bidder $j$ according to $\pi$. 
Our objective is to maximize the revenue from the sale.

The main obstacle is that all the relevant information of each bidder $i$, i.e., their ``type'' $\theta_i = (a_i, d_i, v_i)$, is private information that is unknown to the auction designer, so the auction needs to elicit it from each bidder. However, the bidders can misreport their types and the auction needs to be designed to ensure that they cannot benefit by doing so. Specifically, apart from misreporting her value $v_i$ for the item (which is the standard type of manipulation considered in mechanism design), a bidder can also misreport her arrival and departure times: adopting the original model introduced by \citet{HKP04}, we assume that each bidder $i$ can delay the announcement of her arrival (essentially reporting a delayed arrival time  $\hat{a}_i> a_i$), and she can report a false departure time $\hat{d}_i$ (either earlier or later than her true departure time, $d_i$).
Upon arrival, each bidder $i$ declares a type $\hat{\theta}_i$ (potentially different than $\theta_i$)
and the auction needs to determine who the winner is (i.e., which bidder will be allocated the item), at what time $t$ the item should be allocated to the winner, as well as the amount $p$ 
that the winner should pay.

Apart from the information limitations that the auction faces due to the private nature of the bidders' types, the auction also needs to be implemented in an \emph{online} fashion. This means that if it decides to allocate the item at some time $t$, then this decision is irrevocable, and both this allocation decision and the payment amount requested from the winner can depend only on information regarding bidders with arrival time $a_i\leq t$. In other words, the allocation and payment cannot in any way depend on the types of bidders that have not yet arrived.

If the auction allocates the item to some bidder $i^*$ at some time $t$ for a price of $p$, then this bidder's utility is equal to $v_{i^*}-p$, as long as $t\in [a_{i^*}, d_{i^*}]$, i.e., as long as $i^*$ is active at time $t$. Otherwise, if $i^*$ is allocated the item outside her (real) active time, then she receives no value from it, and her utility is $-p$. All other bidders receive no item and contribute no payment, so their utility is $0$.
An auction is \emph{strategyproof} if for every bidder $i$, truthfully reporting her type is a dominant strategy. This means that no matter what types $\hat{\Theta}_{-i} = (\hat{\theta}_1,\dots, \hat{\theta}_{i-1}, \hat{\theta}_{i+1},\dots , \hat{\theta}_n)$ the other bidders report, the utility of bidder $i$ is maximized if she reports her true type, $\theta_i$. 
To emphasize the added difficulty of achieving strategyproofness in online auctions, relative to static ones, prior work often distinguishes between \emph{value-strategyproofness}, which ensures that bidders will not want to misreport their value, and \emph{time-strategyproofness}, which is the additional requirement to ensure that bidders cannot benefit by misrepresenting their arrival or departure times either.

To evaluate the performance of online auctions with respect to the revenue they extract, prior work focused on a model where a set $I=\{[a_1, d_1], [a_2, d_2],... , [a_n, d_n]\}$ of $n$ arrival-departure intervals and a set $V$ of $n$ values are generated adversarially, and the values of $V$ are then matched to arrival-departure intervals from $I$ uniformly at random. Note that, if the intervals are all non-overlapping, this reduces to the classic \emph{random ordering} model, where the values of the bidders are determined adversarially and the order of their arrival is random. Therefore, our setting generalizes the classic setting of the  ``secretary problem''. We let $\mu(V,I)$ denote the random matching of values to intervals and $\E_{\Theta\sim \mu(V,I)}[\rev(M(\Theta))]$ denote the \emph{expected} revenue of an auction $M$ with respect to this random matching. Also, let $v_{(1)}$ and $v_{(2)}$ denote the highest and second-highest values in $V$, which are important benchmarks since the former is the highest feasible revenue (no bidder would pay more than that) and the latter is the ``offline Vickrey'' benchmark (this corresponds to the amount of revenue that is actually achievable via the classic Vickrey auction in offline settings).

In this work we adopt the learning-augmented framework and study online auctions that
are also equipped with a (potentially very inaccurate) prediction $\prediction$ regarding the highest value, $v_{(1)}$, in $V$. We denote the expected revenue of a auction, $M$, as $\E_{\Theta\sim \mu(V,I)}[\rev(M(\Theta, \prediction))]$ and 
we evaluate the performance of $M$ using its \emph{consistency} and \emph{robustness}. Consistency
refers to the competitive ratio of the expected revenue achieved by the algorithm when the prediction it is provided with is accurate, i.e., whenever $\prediction = \vmax$. The benchmark we use for consistency is the highest value in $V$, often referred to as the \emph{first-best} revenue. Formally:
\[\text{consistency}(M) = \min_{V,I} \frac{\E_{\Theta\sim \mu(V,I)}\left[\rev\left(M\left(\Theta, \vmax\right)\right)\right]}{\vmax}.\]
Robustness refers to the competitive ratio of the expected revenue given an adversarially chosen, inaccurate, prediction. 
The benchmark we use for robustness is the best revenue achievable via any (offline) strategyproof 
auction, i.e., the second highest value 
$v_{(2)}$,
often referred to as the \emph{second-best} revenue. Formally:
\[\text{robustness}(M) = \min_{V, I, \prediction} \frac{\E_{\Theta\sim \mu(V,I)}\left[\rev\left(M\left(\Theta, \prediction\right)\right)\right]}{v_{(2)}}.\]

\section{The $\mech$ Auction}
 We propose the $\mech$ auction, which is parameterized by a value $\alpha\in [0,1]$, with greater values corresponding to higher confidence in the accuracy of the prediction. Our main result in this section shows that for any choice of $\alpha$ this auction achieves $\alpha$-consistency and $(1 - \alpha^2)/4$-robustness, while simultaneously guaranteeing both value-strategyproofness and time-strategyproofness.
 
 The $\mech$ auction considers the bidders based on the order of their departure (i.e., the order of their indices) and comprises three separate phases. 

\begin{enumerate}[leftmargin=*]
    \item During the \textbf{first phase}, the auction observes the values of the first $\lceil\frac{1-\alpha}{2}n\rceil$ bidders to depart (without allocating the item to any of them), aiming to ``learn'' an estimate regarding what a reasonable price for the item may be. If, during this first phase, the auction observes a value that exceeds the predicted maximum, $\prediction$ (implying that the prediction is inaccurate), then, after the first phase is complete, the auction essentially skips the second phase and moves directly onto the third phase.
    If, on the other hand, the first phase does not prove the prediction to be inaccurate, then the auction proceeds to the second phase.
    \item During the \textbf{second phase}, the auction ``tests'' the prediction. Specifically, during this phase (which terminates after $\lfloor \alpha n \rfloor$ more bidders have departed) it asks all active bidders whether they would be willing to pay a price equal to the prediction. If any active bidder is willing to pay this price, then they secure the item and they are guaranteed to pay a price no more than that. The exact payment of bidders who secure the item during the second phase, however, may need to be lower than that to guarantee strategyproofness; we discuss this important subtlety later on. Finally, if none of the $\lfloor \alpha n \rfloor$ bidders is willing to pay a price equal to the prediction during the second phase, then the auction enters its third phase.
    \item During the \textbf{third phase}, the auction offers a take-it-or-leave-it price equal to the highest value observed over all the bidders that have previously departed, and any active bidder can claim the item at that price.
\end{enumerate}

Before going into more detail regarding each of the phases, we note that the auction has a simple description for the special case where no two bidders overlap with respect to their active intervals (i.e., there is just one active bidder at a time). In this case, the auction is a posted price mechanism that posts price $\infty$ to the first $\lceil\frac{1-\alpha}{2}n\rceil$ bidders, then posts price $\max\{\vbest, \prediction\}$ to the next $\lfloor\alpha n\rfloor$ bidders and, finally, if the item remains unsold, it posts price  $\vbest$ to the remaining bidders, where $\vbest$ is the maximum value of bidders who have previously departed. We note that the allocation rule induced by these posted prices is a generalization of the threshold-based algorithm for the classic secretary problem. 

The main challenge, and the main technical portion of our auction, is to handle the cases where there is an overlap between bidders. The time at which the item is allocated, the bidder who receives the item, and the item's price must all be carefully designed to handle bidders who might be active during multiple phases (in particular the second and third phases) and are competing against other bidders. Irrespective of the stage of the auction where the winner $i^*$ is determined,  the item is allocated to $i^*$ at the time of her (reported) departure, $d_{i^*}$, to guarantee time-strategyproofness. If the winner is determined during the third phase, then her final price is the take-it-or-leave-it price that they accepted during this phase. If, on the other hand, the winner is determined during the second phase, the final price needs to be carefully determined in order to guarantee the strategyproofness of the auction. Specifically, if the winner remains active after the transition into the third phase and no other bidder would have claimed the item during the second phase, then the $\mech$ auction may need to reduce  the winner's payment to be equal to the take-it-or-leave-it price that would have been offered during the third phase if we were to remove $i^*$ and simulate the outcome of the auction without them. 

For clarity, we formally present the allocation and the payment rule of the auction separately. Process~\ref{alg:alloc} is the execution of the \textbf{allocation rule}, i.e., it determines who should receive the item. This process maintains a value $\vbest$, corresponding to the maximum value observed among the bidders that have departed so far, and a threshold value $\tau$. If any active bidder has value at least $\tau$, then they can secure the item (tie-breaking using $\pi$ if there are multiple such active bidders). The threshold $\tau$ is $\infty$ during the first phase, then $\max\{\vbest, \prediction\}$ during the second phase, and finally $\vbest$ during the third phase. This process returns the winner $i^*$, if any, and the threshold $\tau$ at which $i^*$ secured the item. Furthermore, to make the formal definition of the payment rule easier, we also let this process return a Boolean variable, ``active-winner,'' which is true only if $i^*$ secured the item right after the transition between two phases. Specifically, this Boolean variable is set to true if the item was secured after the departure of an agent rather than the arrival of one, which implies that the departure caused the transition from one phase to another, leading to a drop in the threshold value, $\tau$, and the winner was already active. \newline

\renewcommand*{\algorithmcfname}{Process}
\begin{algorithm}[H]
\setstretch{1.25}
	\SetKwInOut{Input}{Input}
	\Input{ types $\Theta$ of the $n$ bidders, consistency parameter $\alpha \in [0, 1]$, prediction $\prediction \geq 0$}
$A\leftarrow \emptyset$ \tcp*{the set of active bidders} $L\leftarrow \emptyset$ \tcp*{the set of bidders who have departed}
$\vbest\leftarrow 0$ \tcp*{the maximum value observed so far}
$\tau \leftarrow \infty$ \tcp*{the threshold value to win the item}
$i^* \leftarrow 0$ \tcp*{the agent who receives the item}
active-winner $\gets $ false\; 
\While{$L\neq N$}
{
    \If{some bidder $i$ arrives (tie-break using $\pi$)}
    {
        $A\leftarrow A\cup\{i\}$\;
        \If{$v_i\geq \tau$}
            {$i^* \gets i$ and then break from while-loop}
        
    }
    \ElseIf{some bidder $i$ departs  (tie-break using $\pi$)}
    {
        $A\leftarrow A\setminus\{i\}$ and $L\leftarrow L \cup \{i\}$\;
        $\vbest\leftarrow\max\{\vbest, v_i\}$   \;
        \If{$|L| = \lceil\frac{1-\alpha}{2} n \rceil$  }
        {
            $\tau \leftarrow \max\{\vbest, \prediction\}$ \tcp*{threshold update when entering second phase}  \label{line:first update} 
        }
        \ElseIf{$|L|= \lfloor\frac{1+\alpha}{2} n\rfloor$ }
        {
            $\tau \leftarrow \vbest$ \tcp*{threshold update when entering third phase} \label{line:second update}
        }
        \If{there exists $i\in A$ such that $v_i\geq \tau$}
        {active-winner $\gets$ true\;
        $i^* \gets i$ (tie-break using $\pi$ if needed in choosing $i$) and then break from while-loop \label{line:tiebreaking}\;
        }
    }
}
\Return $i^*$, $\tau$, active-winner
\caption{\textsc{Alloc}: the allocation rule of the $\mech$ auction}
\label{alg:alloc}
\end{algorithm}

\vspace{.5cm}

 Process~\ref{alg:payment} is the execution of the \textbf{payment rule}, i.e., it determines how much the winner, if any, should pay for the item. The price is initially set to be equal to the threshold $\tau$ at which the item was secured and the final price will be no more than that. However, under some circumstances, the price is reduced to guarantee strategyproofness. Specifically, if the winner secured the item during the second phase and remains active during the third phase, they may receive a lower price. In this case, the price is determined by simulating the allocation process without the winning bidder, $i^*$. If the new winner $i'$, in the absence of $i^*$, either i) is not active during the transition into the third phase or ii) loses to $i^*$ in tie-breaking, then the price $p$ is lowered to the threshold $\tau'$ at which $i'$ would have secured the item. Intuitively, if neither of these two conditions holds and we did not offer $i^*$ the reduced price, then $i^*$ could report a value of $\tau'$ instead of her true value and secure the item at that lower price right after the transition into the third phase. \newline

\begin{algorithm}[H]
\setstretch{1.1}
\SetKwInOut{Input}{Input}
\Input{the output of the allocation rule: $i^*$, $\tau$, active-winner}
  $p \gets \tau$ \tcp*{initialize the price to $\tau$, the threshold at which $i^*$ claimed the item}
\tcp{If the winner is determined in second phase and remains active in third phase}
\If{$\tau = \prediction$ and $i^* > \lfloor\frac{1+\alpha}{2}n\rfloor$}{
     $i', \tau', \text{active-winner}' \leftarrow  \textsc{Alloc}(n-1, \Theta_{-i^*}, \alpha, \prediction)$\label{line:rerun}\tcp*{simulate allocation without $i^*$} 
    \tcp{If $i'$ is inactive when entering third phase or loses in tie-breaking to $i^*$}
    \If{$\text{active-winner}' =$ false or $i^* \pref i'$ }{
        $p \gets \tau'$ \tcp*{update the price to $\tau'\leq \tau$} \label{line:loserdefine}
    }
}

\Return{$p$}
\caption{the payment rule of the $\mech$ auction}
\label{alg:payment}
\end{algorithm}

\vspace{.2cm}

\begin{observation}
The $\mech$ auction can be implemented in an online fashion.
\end{observation}
\begin{proof}
It is easy to verify that the allocation rule is online implementable since the auction maintains a threshold $\tau$ at any point and decides the winner when some active bidder's value is above the threshold; this requires no future information. We now argue that the payment rule is online implementable as well. Crucially, note that the winner is allocated the item at the time of their departure, so all we need to argue is that the price that they need to pay can be determined at that point. To verify this fact, note that if $i^*$ is not active during the transition from the second phase to the third phase, then her price is just $\tau$. If, on the other hand, $i^*$ is active during that transition, then the auction can also check the value of any other bidder that is also active up to that transition to determine $i', \tau'$ and active-winner$'$, without needing to simulate any portion of the allocation rule beyond the departure of $i^*$.
\end{proof}

Our main result in this section shows that the $\mech$ not only guarantees value- and time-strategyproofness, but it also achieve a non-trivial tradeoff between robustness and consistency.

\begin{theorem} \label{thm:mechanismmaster}
\mech\ is a value-strategyproof and time-strategyproof online auction that, given any parameter $\alpha\in \range$, simultaneously guarantees $\alpha$-consistency and $\frac{1 - \alpha^2}{4}$-robustness.\end{theorem}
In Section~\ref{sec:revenue}, we prove the consistency and robustness guarantees achieved by our auction. In Section~\ref{sec:strategyproof}, we show that it is strategyproof. Finally, in Section~\ref{sec:error}, we give an extension of our auction that achieves revenue guarantees as a function of the prediction quality.  For presentation purposes we use $\first$ to denote $\frac{1 - \alpha}{2}n$ and $\second$ to denote $\frac{1 + \alpha}{2}n$ in the following analysis.

\subsection{Revenue Guarantees}
\label{sec:revenue}
In this section we analyze the performance of our auction in terms of consistency and robustness. We focus on the values of $\alpha$ in the set $\range=\{\alpha \in [0,1]: \alpha n \in \mathbb{N} \text{ and } \frac{1-\alpha}{2}n \in \mathbb{N}\}$ which make $\alpha n$ and  $\frac{1-\alpha}{2}n$ (the number of bidder departures in the first two phases) are integral. 

\begin{lemma}\label{lem:consistency}
For any $\alpha \in \range$, \mech\ is $\alpha$-consistent.
\end{lemma}
\begin{proof}
Assume that the prediction is correct, i.e., $\prediction = \vmax$. First, observe that $\prediction \geq v_i$ for all $i\in [n]$, so $\tau = \prediction = \vmax$ during Phase 2, and no one would be above the threshold besides the highest value bidder. Additionally, note that no bidder is allocated the item in the first phase. Then the auction would be able to extract revenue $\vmax$ if the highest value bidder is allocated the item during Phase 2 and pays the price $\vmax$. Let $i^*$ be the bidder with the highest value; it is sufficient to guarantee the aforementioned outcome if the departure time is between the $\first$-th departure and the $\second$-th departure, i.e., $i^* \in [\first+1, \second]$, which, based on our random-ordering assumption, occurs with a probability of $\frac{\second - \first}{n} = \alpha$. 
\end{proof}
\begin{lemma}\label{lem:robustness}
For any $\alpha \in \range$, \mech\ is $\frac{(1-\alpha^2)}{4}$-robust.
\end{lemma}
\begin{proof} Recall that robustness is measured relative to $v_{(2)}$. We now consider the cases where we under-predict and over-predict separately. In the first case, we have $\prediction > \vmax$. Since $\tau = \prediction > v_i$ for any bidder $i$, we have that bidders can only be above the threshold in phase 3. In this case, by the payment rule, it is sufficient to extract $v_{(2)}$ if $\tau = v_{(2)}$, which requires the second highest bidder to be amongst the first $\second$ bidders to depart and bidders with the highest value are amongst the last $n - \second$ to depart. Then with probability $\frac{\second}{n} \frac{n - \second}{n - 1} = \frac{1 - \alpha^2}{4} + O(\frac{1}{n})$, revenue $v_{(2)}$ is extracted.

In the second case, we have $\prediction \leq \vmax$. Observe that in the under-predicted case, the prices in phases 2 and 3 are guaranteed to fall in the range $[v_{(2)}, \vmax]$ if the second highest bidder is amongst the first $\first$ bidders to depart. Bidders with the highest value may be allocated the item if they are present at any point in phases 2 and 3, meaning they are amongst the last $n - \first$ to depart. Then with probability $\frac{\first}{n} \frac{n - \first}{n - 1} = \frac{1 - \alpha^2}{4} + O(\frac{1}{n})$, revenue $v_{(2)}$ is extracted.
\end{proof}
\subsection{Strategyproofness}\label{sec:strategyproof} 
In this section we show our auction is both value-strategyproof and time-strategyproof. All the missing proofs are deferred to Appendix~\ref{app:strategyproof}. We first show that the bidders with the first $\first$ departure times have no incentive to misreport.
\begin{lemma}\label{lem:firstphase}
Consider some bidder $i$ and any $\hat{\Theta}_{-i}$. If bidder $i$'s true departure time is in the first $\first$, then bidder $i$ has no incentive to misreport her type $\theta_i$.
\end{lemma}
\begin{proof}
First note that if such bidder $i$ reports her type truthfully, she won't receive the item since $\tau = \infty$ before the departure of the $\first$-th bidder. Let $\hat{d}$ be the reported departure time of the $\first +1$-th bidder based on the departure schedule. The only way for bidder $i$ to possibly obtain the item is to report a later departure time, denoted as $\hat{d}_i$, such that $\hat{d}_i > \hat{d} \geq d_i$. However, the auction allocates at her reported departure time $\hat{d}_i$, which falls outside her active time. Based on our assumption, she receives no value from the item. Such bidder therefore has no incentive to change her type to obtain the item.
\end{proof}

We now make the following observation: if there exists a value in the first $\first$ bidders that is weakly more than the prediction, then the price the winner pays is fixed.
\begin{restatable}{rLem}{obssinglethreshold}
\label{obs:singlethreshold}
Let $\bestsofar{\first}$ be the maximum value of the first $\first$ departed bidders. If in Line~\ref{line:first update}, $\tau = \bestsofar{\first}$, then the price winners pays is $p = \bestsofar{\first}$. 
\end{restatable}

We refer to the above scenario as the \emph{single-threshold case} since $\tau$ is effectively only updated once (from $\infty$ to $\bestsofar{\first}$), analogously, we refer to the other scenario as the \emph{two-threshold case}. We note that only the first $\first$ bidders can define $\bestsofar{\first}$ therefore they (together with the prediction) decide which case the rest of the bidders are in. We first show that the rest of the bidders have no incentive to lie in the single-threshold case. 

\begin{restatable}{rLem}{lemvalbidderfirst}
 \label{cla:valbidderfirst}
 Consider some bidder $i$, and any $\hat{\Theta}_{-i}$ that results in the single-threshold case and let $\tau$ be the thresholds defined in Line~\ref{line:first update}. Bidder $i$  has no incentive to misreport her type $\theta_i$.
\end{restatable}

We now discuss the more involved case, the two-threshold case. For the ease of presentation, we will denote the threshold defined in Line~\ref{line:first update} as $\tau_1$ and the threshold defined in Line~\ref{line:second update} as $\tau_2$. Note that $\tau_2 < \tau_1$ ($\tau_2 = \tau_1$ is equivalent to the single-threshold case). We first show that winners in this case cannot manipulate the price via misreporting.

\begin{lemma}\label{lem:winner2}
Consider some bidder $i$, and any $\hat{\Theta}_{-i}$ that results in the two-threshold case. If $i$ is the winner with her true type $\theta_i$, then bidder $i$ has no incentive to misreport her type $\theta_i$. 
\end{lemma}
\begin{proof}
Consider the two possible prices the bidder is paying, $\tau_1$ and $\tau_2$. If bidder $i$ wins and pays $\tau_2$ (the cheaper one), she has no incentive to deviate as it is the best outcome. Consider the cases where bidder $i$ wins with price $\tau_1$. By the payment rule, she either left before the $\tau$ update (she is ranked before $\second$ with respect to departure time) or there must exist a bidder $i'$ that wins in Line~\ref{line:rerun} of the payment rule. Consider the three possible cases below:

\textbf{Case one:} $i \leq \second$. In this case the bidder $i$ departs before the threshold drops to $\tau_2$. To possibly get the lower price, she has to report a departure time $\hat{d}_i > \hat{d}_{\second} \geq d_i$. However, the auction allocates at her reported departure time $\hat{d}_i$ which falls outside her active time. Based on our assumption, she receives no utility from the item. 

\textbf{Case two:} $\tau' = \tau_1$, so some other bidder $i'$ is above the threshold $\tau_1$ before it drops. In this case, the winner cannot reduce the price since the existence of such a bidder is independent of her report.

\textbf{Case three:} $\checkp = \text{true}$ and $i' \pref i$. This case happens when $i'$ becomes available right after the threshold drops to $\tau_2$. Since both the tie-breaking rule and the existence of bidder $i'$ are independent of bidder $i$'s report, bidder $i$ can't get a better price via misreporting.
\end{proof}

We now show that the losing bidders with values below $\tau_2$ and above $\tau_1$ cannot benefit from misreporting as well.  The proof is almost identical to the proof of Lemma~\ref{cla:valbidderfirst}. 

\begin{restatable}{rLem}{lemvalbiddertwo}
\label{cla:valbidder2}
Consider some bidder $i$, and any $\hat{\Theta}_{-i}$ that results in the two-threshold case and let $\tau_1$ and $\tau_2$ be the thresholds defined in Line~\ref{line:first update} and Line~\ref{line:second update} respectively.  If bidder $i$ has a value $v_i \leq \tau_2$ or  $v_i \geq \tau_1$, she has no incentive to misreport her type $\theta_i$.
\end{restatable}


The next lemma shows that losing bidders with values  between the two thresholds have no incentive to lie.
\begin{lemma}\label{cla:midvalbidder2}
Consider some bidder $i$, and any $\hat{\Theta}_{-i}$ that results in the two-threshold case and let $\tau_1$ and $\tau_2$ be the thresholds defined in Line~\ref{line:first update} and Line~\ref{line:second update} respectively.  If bidder $i$ has a value $\tau_1 > v_i > \tau_2$ and she is not the winner, she has no incentive to misreport her type $\theta_i$.
\end{lemma}
\begin{proof}
Consider a bidder $i$ 
with a value such that $\tau_1 > v_i > \tau_2$ who is not the winner. The only outcome that is strictly better for her is winning the item with a price of $\tau_2$. For the rest of the proof we show that such an outcome is not obtainable by such bidders through misreporting. Since $i$ is not the winner, there must be another bidder $i^*$ who either has a value above the threshold before $i$ or $i^*$ is above the threshold at the same time as $i$ but $i^* \pref i$. Since $i$'s value is less than $\tau_1$, she can only be above the threshold  
after the departure of the $\second$-th bidder if she reports her true type.

\textbf{Case one:} If bidder $i^*$ is above the threshold before bidder $i$, two scenarios are possible. First, bidder $i^*$ is above the threshold before the arrival time of bidder $i$, in which case there is no way bidder $i$ can misreport and win the item, as we assume bidders cannot report an arrival time earlier than their actual arrival time. If bidder $i^*$'s value is not above the threshold before bidder $i$'s arrival time but is earlier than when $v_i$ is above the threshold, it must be that bidder $i$'s arrival time is before the threshold drops to $\tau_2$, and bidder $i^*$ is above the threshold $\tau_1$. In this case, bidder $i$ indeed can report a value $\hat{v}_i \geq \tau_1$ to win the item. However, due to the presence of bidder $i^*$, the price she needs to pay is $p = \tau_1$ by Line~\ref{line:loserdefine} of the payment subroutine, implying that bidder $i$ would reduce her utility by misreporting in this way.

\textbf{Case two:} If bidder $i$ loses to bidder $i^*$ in the tie-breaking, this can only occur when the threshold drops to $\tau_2$\footnote{In all of the cases, we consider the bidders with respect to some ordering even if they arrive at the same time.}. The only way for bidder $i$ to win the item is if she has an arrival time before the threshold drop and reports a value $\hat{v}_i \geq \tau_1$ to avoid the tie-breaking. However, due to the presence of bidder $i^*$, we would get $\checkp = \text{true}$ and $i' \pref i$, making $p = \tau = \tau_1$. Even if bidder $i$ obtains the item, her utility is non-positive.
\end{proof}

We are now ready to show the main Lemma of the subsection
\begin{lemma}\label{lem:strategyproofness}
\mech\ is both value-strategyproof and time-strategyproof.
\end{lemma}
\begin{proof}
From Lemma~\ref{cla:valbidderfirst} we get that in the single-threshold case, no bidder $i$ has incentive to misreport her type $\theta_i$. Combining Lemma~\ref{lem:winner2}, \ref{cla:valbidder2}, and \ref{cla:midvalbidder2}, we get that in the two-threshold case, no bidder $i$ has incentive to misreport her type $\theta_i$.
\end{proof}
\subsection{The \emech\ Auction}\label{sec:error}

We now show that our auction can be easily extended to achieve an improved revenue guarantee not only when the prediction is perfectly accurate, but even when it is approximately accurate. 
Given a prediction $\prediction$ regarding the maximum bidder value $\vmax$, we use $\quality(\prediction, \vmax)$, or just $\quality$, to capture the prediction quality, defined as the relative under- or over-prediction: $\quality = \min \left\{\frac{\prediction}{\vmax}, \frac{\vmax}{\prediction}\right\}.$
Note that $\quality \in [0,1]$ and that higher values of $\quality$ correspond to better predictions. 

We start by describing  the $\emech$ \ auction, which is an extension of the $\mech$ \ auction. This auction takes as input an additional parameter  $\gamma \in [0,1]$ called the error-tolerance parameter, whose value is chosen by the auction designer. The only change from $\mech$ \ to $\emech$ \ is that  Line~\ref{line:first update} is changed from $\tau = \max\{\vbest, \prediction\}$ to $\tau = \max\{\vbest, \gamma \cdot \prediction\}$. The main result for the $\emech$ \ auction is that when the prediction quality $q$ is at least as high as the error-tolerance $\gamma$, then the auction achieves a revenue guarantee of  $\max\{\alpha\gamma\quality \cdot  v_{(1)}, \frac{1 - \alpha^2}{4} \cdot v_{(2)}\}$. Thus, in that case, a competitive ratio of $\alpha\gamma\quality$ is guaranteed against the first-best revenue benchmark  $v_{(1)}$, even if the prediction is not perfectly accurate. In addition, a competitive ratio of $\frac{1 - \alpha^2}{4}$ against the second-best revenue benchmark $v_{(2)}$ is always  maintained. Note that if we let $\gamma=1$, the auction reduces to $\mech$ and so do the corresponding robustness and consistency guarantees. We defer the proof of the following theorem to  Appendix~\ref{app:error}.

\vspace{-.1cm}

\begin{restatable}{rThm}{thmError}
\label{thm:error}
\emech\ is a value-strategyproof and time-strategyproof online auction that, given parameters $\alpha\in \range$ and $\gamma \in [0,1]$, and a prediction with unknown quality $\quality$, achieves
\[\E_{\Theta \sim \mu(V, I)} [\rev(M(\Theta, \prediction))]\geq \begin{cases}
    \max\left\{\alpha\gamma\quality \cdot v_{(1)}, \frac{1 - \alpha^2}{4} \cdot v_{(2)}\right\} & \text{if $\quality \geq \gamma$,} \\
    \frac{1 - \alpha^2}{4}  \cdot v_{(2)} & \text{if $\quality < \gamma$.}
\end{cases}\]
\end{restatable}

\section{Impossibility Results}

In this section, we complement our positive results with impossibility results both for the original online mechanism design setting (without predictions) and for the one with predictions. To provide some perspective,  we first briefly discuss two related problems: the classic secretary problem and the problem of designing online auctions to maximize social welfare.

In the secretary problem, a sequence of agents arrive in random order, and the goal is to maximize the probability that the highest-value one is chosen. Each time a new agent appears, the mechanism observes that agent's ranking relative to all the previous ones and then irrevocably decides whether to accept or reject that agent. Secretary algorithms have two significant restrictions compared to online auctions: (i) they are rewarded only if they accept the highest-value agent, whereas online auctions can extract value  from any agent with a positive value, and (ii) their decisions depend only on ordinal information regarding the agents, whereas online auctions observe the previous agents' cardinal values and the auction's pricing decisions can be an arbitrary function of the values it has observed. Therefore, the design space for online auctions is much richer (due to (ii)) and proving impossibility results regarding their performance is much more demanding (due to both (i) and (ii)).


A problem that is more closely related to ours is the design of online auctions for a single good aiming to maximize the social welfare, i.e., the value of the agent who wins the good. Just like the secretary problem, the ideal outcome would be to choose the highest-value agent, but the auction observes more information and can extract welfare even if other high-value agents are chosen. For this problem,~\citet{HKP04} provided a mechanism that achieves a $1/e$ approximation and proved it is impossible to achieve better than $1/2$. To the best of our knowledge, a tight impossibility result of $1/e+o(1)$ was only recently achieved in the EC 2019 best paper by \cite{correa}. This result is actually shown even in a setting where the values are generated i.i.d.\ and it is significantly more technical, relying on infinite hypergraphs and Ramsey's theorem.

Proving impossibility results for revenue maximization introduces additional challenges relative to social welfare maximization. First, note that all the tight impossibility results for social welfare did not need to leverage strategyproofness: they hold even for non-strategyproof mechanisms, so strategyproofness comes for free. However, this is not the case for revenue maximization, where the strategyproofness constraint needs to be used in order to achieve even the $2/3$ impossibility result in~\citet{HKP04}. Second, the performance of an auction with respect to revenue depends not only on who wins the good, but also on the price that the winner pays. This poses a significant challenge, since one cannot easily rule out the existence of (rather unrealistic) mechanisms that happen to ``guess'' a high price as an arbitrary function of the observed values and extract a lot of revenue from the highest-value bidder. As a result, there remains a gap between the best-known approximation of $1/4$ and the best-known impossibility of $2/3$~\citet{HKP04}.

\paragraph{Overview of this section.}
In Section~\ref{sec:lbrobust}, we prove that
the competitive ratio of $1/4$  is, in fact, optimal for a large family of auctions whose pricing rule can be
an arbitrary function of the values observed in the past, but no more than the maximum seen previously.
The proof relies on an LP formulation of the revenue maximization problem and duality. For the
learning-augmented setting, we prove in Section~\ref{sec:lbtrade-off}  that the robustness-consistency trade-off achieved by our
Three-Phase auction is optimal within a broad family of learning-augmented online auctions.  The previous LP approach requires a history-independence property that no longer holds
in the setting with predictions. For this second impossibility result, we instead use an interchange
argument to show the existence of an optimal auction with a three-phase structure identical to our
auction.

\paragraph{Some preliminaries.}
We prove our impossibility results by restricting our attention to instances with no overlapping active intervals, i.e., at each time step $i \in [n]$, there is a single active bidder $i$ such that $a_i = d_i = i$, and we refer to bidder $i$ as the bidder who arrives and departs at time $i \in [n]$). When the intervals are disjoint the agents cannot benefit from misreporting their arrival and departure times, so our impossibility results hold even for mechanisms that do not satisfy time-strategyproofness. We use $V = \{v_{(1)}, \ldots, v_{(n)}\}$ to denote an instance with $n$ values where $v_{(i)}$ is the $i$th highest value. A random matching $\mu \sim \mathcal{S}_n$, where  $\mathcal{S}_n$  denotes the set of permutations over $n$ indices, between values $V$ and intervals $I = [n]$ is a random ordering over $V$, which we denote by $\mathbf{v} = (v_1, \ldots, v_n)$. We abuse notation and let $v_i = v_{(\randorder^{-1}(i))}$ be the value of the active agent at time step $i.$   We define for any $i > 0$ the notation $v_{\max}^{<i} = \max\{v_1,\dots,v_{i-1}\}$ to represent the maximum value seen before $i$. Finally, we note that strategyproof online auction can be implemented as a posted price auction with prices that depend only on previously seen values~\cite{lavi2000competitive}.
\subsection{Tight Impossibility Result for Setting without Predictions}\label{sec:lbrobust}
In this section we present a $1/4$ impossibility result  for revenue maximization (without predictions)  for a broad family of auctions, which matches the $1/4$ competitive ratio achieved by \citet{HKP04} for revenue. The proof uses similar tools as the result in~\cite{correa} that implies the $1/e$ impossibility result for welfare maximization. The main difference is that, in~\cite{correa}, the performance of an allocation rule is measured by the probability of allocating to the maximum valued agent, whereas for revenue maximization, it is measured by the probability of allocating to the maximum valued agent, while also having  previously seen the second highest valued agent. In particular, our proof also uses linear programming and duality, but with a different linear program whose objective is this different performance metric that leads to a $1/4$, instead of $1/e$, bound. 

We note that \citet{gravin2023online} show that for any ordinal stopping problem (i.e. a problem whose objective only depends on ordinal information  about the agents) in the online random-order setting, there is only a small gap between the performance of the optimal  cardinal algorithm (i.e. an algorithm that can depend on the agents' values) and the performance of the optimal ordinal algorithm (i.e. an algorithm that  only depends on ordinal information  about the agents).  Since our revenue maximization problem is not an ordinal stopping problem (the objective itself heavily depends on cardinal preferences), their result does not apply. Although there are intermediary steps of our proof that consider an ordinal problem, their reduction cannot be used for these intermediate  steps either because our proof requires an arbitrarily large gap between the second and third largest value of a hard instance for this intermediate ordinal problem and it is unclear whether the construction in \cite{gravin2023online} satisfies this requirement.


All missing proofs are deferred to Appendix~\ref{app:lbrobust}.  Note that throughout this section, when we consider a set of values, we implicitly assume that they are pairwise distinct. Below is the family of auctions for which we prove this result.

\begin{definition}
    An  auction is in the family of Up to Max-Previously-Seen auctions $\mathcal{M}_{u}$ if for all  $i\in[n]$ and   bids $v_1,\dots,v_{i-1}$, it posts a price $p_i(v_1,\dots,v_{i-1}) \in [0,\max\{v_1,\dots,v_{i-1}\}]\cup\infty$ to agent $i$.
\end{definition}

 Up to Max-Previously-Seen auctions capture all strategyproof auctions except those that post bounded prices that are larger than the maximum value seen previously. 

\begin{theorem}\label{thm:atmostmsfbound}
    For any $n\geq 2$ and any Up to Max-Previously-Seen auction $M\in \mathcal{M}_{u}$, there is an instance $V = \{v_{(1)}, \dots, v_{(n)}\}$ such that $\E_{\mathbf{v} \sim \mu(V)}[\rev(M(\mathbf{v}))] \leq \left(\frac{1}{4}+\frac{7}{n}\right)v_{(2)}$. 
\end{theorem}

This result matches the robustness achieved by the auction of~\cite{HKP04} and by our \textsc{Three-Phase} auction with $\alpha = 0$ for revenue maximization.

\subsubsection{The main lemma} The main lemma to prove Theorem~\ref{thm:atmostmsfbound} is regarding the allocation rule of any auction. With no overlapping active intervals, an allocation rule is simply a stopping rule.

\begin{definition}
    A stopping rule $\stoprule$ is a collection of $n$ functions $s_i(v_1,\dots,v_i)$, for $i \in [n]$, that define the probability of allocating the item to bidder $i$, conditioned on  $v_1,\dots,v_i$ being the values of the first $i$ bidders and not having allocated the item at time $j < i$. The stopping time $\stoptime$ of a stopping rule $\stoprule$ is the time at which the item is allocated. Thus,  $s_i(v_1,\dots,v_i) = \PP(\stoptime = i|v_1,\dots,v_i,\stoptime \geq i)$.
\end{definition}

Note that here, and throughout this section, expectations and probabilities are over the random order $\mathbf{v} \sim \mu(V)$ of values and the randomness in the stopping rule, unless otherwise stated or implied through conditioning. The main lemma shows that for any stopping rule with stopping time $\stoptime,$  the expected maximum value $\E[v_{\max}^{<\stoptime}]$ seen before time $\stoptime$ gets a contribution $\E[v_{\max}^{<\stoptime}|v_{\stoptime}\geq v_{\max}^{<\stoptime}]\cdot\PP(v_{\stoptime}\geq v_{\max}^{<\stoptime})$ from the event that  $v_{\stoptime}$ is the maximum value seen so far that  is a $1/4 + o(1)$ approximation to the second highest value $v_{(2)}.$  

\begin{lemma}\label{thm:1/4}
For any stopping rule $\stoprule$ with stopping time $\stoptime$ over $n \geq 2$ values, there is  an instance $V = \{v_{(1)},\dots,v_{(n)}\}$ such that $\E[v_{\max}^{<\stoptime}|v_{\stoptime}\geq v_{\max}^{<\stoptime}]\cdot\PP(v_{\stoptime}\geq v_{\max}^{<\stoptime}) \leq \left(\frac{1}{4} + \frac{4}{n}\right)v_{(2)}.$
\end{lemma}

Note that, for a stopping rule $\stoprule$, $\E[v_{\max}^{<\stoptime}|v_{\stoptime}\geq v_{\max}^{<\stoptime}]\cdot\PP(v_{\stoptime}\geq v_{\max}^{<\stoptime})$ is equal to  the revenue achieved by  posting, at every time step $i$, price $v_{\max}^{<i}$ with probability $s_i(v_1,\dots,v_i)$ and price $\infty$ with probability $1- s_i(v_1,\dots,v_i)$ to agent $i$.
 We will then prove that for any auction in $\mathcal{M}_{u}$, there exists an instance for which the revenue gained by either stopping  not on the highest value or stopping while having not seen the second highest value is small. 
 
\paragraph{Overview of the proof of Lemma~\ref{thm:1/4}.} 
We  prove the main lemma as follows:
\begin{enumerate}
    \item  For any stopping rule $\stoprule$, we show in Section~\ref{sec:orderoblivious} that there is an order-oblivious stopping rule, i.e., a stopping rule that is independent of the order of the values previously seen, with equal $\E[v_{\max}^{<\stoptime}|v_{\stoptime}\geq v_{\max}^{<\stoptime}]\cdot\PP(v_{\stoptime}\geq v_{\max}^{<\stoptime})$ on any instance (Lemma~\ref{lem:orderoblivious}).
    \item For the new order-oblivious stopping rule, we prove the existence of an infinitely large support of natural numbers on which the stopping rule is approximately value-oblivious (Lemma~\ref{lem:valoblivfromorder}), meaning decisions only depend on the relative rank of values.
    \item In Section~\ref{sec:LP}, we bound   the expected maximum value  seen before the stopping time of all such stopping rules using an LP and duality (Lemma~\ref{lem:LPbound}).
    \item Finally, we combine these results to complete the proof of this result in Section~\ref{sec:finishlemma}.
\end{enumerate}


\subsubsection{Reduction to Approximately Value-oblivious Stopping Rules}\label{sec:orderoblivious} 
\begin{definition}
    A stopping rule $\stoprule$ is \emph{order-oblivious} 
 if for all time $i\in [n]$, values $w_1,\dots,w_i\in \mathbb{R}_{\geq 0}$, and permutations $\randorder\in \mathcal{S}_{i-1}$, $s_i(w_1,\dots,w_i) = s_i(w_{\randorder(1)},\dots,w_{\randorder(i-1)},w_i)$.
 \end{definition}
The following lemma finds for any stopping rule $\stoprule$, a corresponding order-oblivious stopping rule $\stoprule'$ that achieves the same expected maximum value seen before the stopping time as $\stoprule$. This allows us to restrict our attention to order-oblivious stopping rules.
\begin{restatable}{rLem}{lemorderoblivious}
\label{lem:orderoblivious}
    For any  $\stoprule$ with stopping time $\stoptime$, there exists an $\stoprule'$ with stopping time $\stoptime'$ such that (1) $\stoprule'$ is order-oblivious, (2) for any instance $V$ and step $i \in [n]$,  if $\PP(\stoptime = i) = 0$, then $\PP(\stoptime' = i) = 0$, and (3) for any instance $V$, $\E[v_{\max}^{<\stoptime'}|v_{\stoptime'}\geq v_{\max}^{<\stoptime'}]\cdot\PP(v_{\stoptime'} \geq v_{\max}^{<\stoptime'}) = \E[v_{\max}^{<\stoptime}|v_{\stoptime}\geq v_{\max}^{<\stoptime}]\cdot\PP(v_{\stoptime}\geq v_{\max}^{<\stoptime}).$
\end{restatable}

We then wish to find supports on which they behave approximately like pairwise comparison rules, which depend only on relative rank of the values seen before.
\begin{definition}
    Let $\varepsilon>0$ and $Y\subseteq \mathbb{N}$. A stopping rule $\stoprule$ is $\varepsilon$-value-oblivious on an instance $Y$ if for all $i\in[n]$, there exists $y_i\in[0,1]$ such that, for all values $v_1,\dots,v_i\in Y$ with $v_i>\max\{v_1,\dots,v_{i-1}\}$, it holds that $s_i(v_1,\dots,v_i)\in[y_i-\varepsilon,y_i+\varepsilon)$.
\end{definition}

The following lemma shows, for any order-oblivious stopping rule $\stoprule$ and arbitrarily small $\varepsilon > 0$, the existence of an infinitely large support on which $\stoprule$ is $\varepsilon$-value-oblivious. 
The proof of Lemma~\ref{lem:helper} is inside the proof of Lemma 2 in~\citet{correa}.

\begin{lemma}[\citet{correa}]\label{lem:helper}
    For any order-oblivious stopping rule $\stoprule$ and any $\varepsilon > 0$, there exists an infinite set $Y\subseteq \mathbb{N}$ such that $\stoprule$ is $\varepsilon$-value-oblivious on Y.
\end{lemma}
By combining Lemmas~\ref{lem:orderoblivious} and~\ref{lem:helper}, we show that restricting our attention to $\varepsilon$-value-oblivious stopping rules is sufficient to bound $\E[v_{\max}^{<\stoptime}|v_{\stoptime}\geq v_{\max}^{<\stoptime}]\cdot\PP(v_{\stoptime}\geq v_{\max}^{<\stoptime})$ for all stopping rules.
  \begin{lemma} 
  \label{lem:valoblivfromorder}
     Let $\varepsilon > 0$. For any stopping rule $\stoprule$ with stopping time $\stoptime$, there exists  stopping rule $\stoprule'$ with stopping time $\stoptime'$ that is $\varepsilon$-value-oblivious on an infinite set $Y \subseteq \mathbb{N}$ such that, over any instance $V$,
     $\E[v_{\max}^{<\stoptime'}|v_{\stoptime'}\geq v_{\max}^{<\stoptime'}]\cdot\PP(v_{\stoptime'}\geq v_{\max}^{<\stoptime'}) = \E[v_{\max}^{<\stoptime}|v_{\stoptime}\geq v_{\max}^{<\stoptime}]\cdot\PP(v_{\stoptime}\geq v_{\max}^{<\stoptime}).$
  \end{lemma}
\begin{proof}
Let $\stoprule$ be an arbitrary stopping rule with stopping time $\stoptime$. By Lemma~\ref{lem:orderoblivious} we know there exists an order-oblivious stopping rule $\stoprule'$ with stopping time $\stoptime'$ such that $\E[v_{\max}^{<\stoptime'}|v_{\stoptime'}\geq v_{\max}^{<\stoptime'}]\cdot\PP(v_{\stoptime'}\geq v_{\max}^{<\stoptime'}) = \E[v_{\max}^{<\stoptime}|v_{\stoptime}\geq v_{\max}^{<\stoptime}]\cdot\PP(v_{\stoptime}\geq v_{\max}^{<\stoptime})$ for every instance, and by Lemma~\ref{lem:helper} there exists infinite $Y\subseteq \mathbb{N}$ such that $\stoprule'$ is $\varepsilon$-value-oblivious on $Y$.
\end{proof}

\subsubsection{LP Bound} \label{sec:LP}
For the  stopping time $\stoptime$ of  stopping rules that are $0$-value-oblivious over some $Y\subseteq \mathbb{N}$, we  bound  $\PP(\randorder(1) = \stoptime, \randorder(2) < \stoptime)$ by using the following linear program and its dual.
\begin{equation*}
    \begin{split}
        LP: \max \; & \sum_{i = 1}^n\frac{(i - 1)i}{(n-1)n}x_i \\
        s.t. \; & i\cdot x_i \leq 1-\sum_{j = 1}^{i - 1}x_j \; \forall i\in[n] \\
        & x_i \geq 0 \; \forall i \in[n]
    \end{split}
    \hspace{2cm}
    \begin{split}
        DP: \min \; & \sum_{i = 1}^ny_i \\
        s.t. \; & i\cdot y_i \geq \frac{(i - 1)i}{(n - 1)n} - \sum_{j = i + 1}^ny_j \; \forall i\in[n] \\
        & y_i \geq 0 \; \forall i\in[n]
    \end{split}
\end{equation*}

This construction is conceptually similar to the one in~\cite{BJS08}, which bounds the probability of stopping on the highest value. The primal variables $x_i$ represent $\PP(\stoptime = i, v_i \geq v_{\max}^{<i})$ for each time step $i\in[n]$ and we wish to find the maximum value of $\PP(\randorder(1) = \stoptime,\randorder(2) < \stoptime)$. The first constraint thus enforces that $x_i$ cannot exceed the probability that $v_i$ is the highest so far and the stopping rule has not already stopped prior to $i$. The second constraint enforces nonnegativity. Then in the objective function, the coefficient of each $x_i$ represents the probability that $v_{(1)}$ and $v_{(2)}$ are amongst the first $i$ values. Thus, stopping at $i$, given that $v_i$ is the highest so far, corresponds precisely to the event $\{\randorder(1) = \stoptime, \randorder(2) < \stoptime\}$.

For a $0$-value-oblivious stopping on instance $Y$, we let $s_i = s_i(v_1, \ldots, v_i)$,  be the probability of stopping on $v_i$ assuming that $v_i$ is the maximum so far, i.e., $v_i \geq \max\{v_1, \ldots, v_{i-1}\}$, and that $v_1, \ldots, v_{i} \in Y.$ Observe that stopping when $v_i$ is not the maximum so far does not increase $\PP(\randorder(1) = \stoptime, \randorder(2) < \stoptime)$. Then we can restrict our attention to stopping rules that do not stop unless the current value is the maximum so far. We can then define for any such 0-value-oblivious stopping rule variables $x_i = \PP(\stoptime = i)$, and we first show that these form a solution to the LP.

\begin{restatable}{rLem}{lemalgtoLPfeas}\label{lem:algtoLPfeas}
    Let $\stoprule$ be a $0$-value-oblivious stopping rule with stopping time $\stoptime$ on an instance $Y$, and $s_i(v_1,\dots,v_i) = 0$ if $v_i < \max\{v_1,\dots,v_{i-1}\}$. Then if $x_i = \PP(\stoptime = i)$ for all $i\in [n]$, $x_1,\dots,x_n$ is a feasible solution to the LP.
\end{restatable}

 The following  lemma allows us to remove the conditioning on the absolute order event $\{\randorder(1) = i, \randorder(2) < i\}$ from the probability of stopping at $i$ if $v_i$ is already known to be the maximum so far.

\begin{restatable}{rLem}{lemtauequalsi}\label{lem:tauequalsi}
    Let $\stoprule$ be a $0$-value-oblivious stopping rule with stopping time $\stoptime$ on an instance $Y$ such that, for all $i \in [n]$, $s_i(v_1,\dots,v_i) = 0$ if $v_i < \max\{v_1,\dots,v_i\}$. Over instance $Y$, we have that, for any $i\in [n]$, $\PP(\stoptime = i|v_i \geq v_{\max}^{<i}, \randorder(1) = i, \randorder(2) < i) = \PP(\stoptime = i|v_i \geq v_{\max}^{<i})$.
\end{restatable}

Lemma~\ref{lem:algtoLPobj}  proves that the objective value of the LP with respect to our variables $x_i$ bounds the probability of stopping on the highest value $v_{(1)}$ having seen the second highest $v_{(2)}$.

\begin{lemma}\label{lem:algtoLPobj}
    Let $\stoprule$ be a $0$-value-oblivious stopping rule with stopping time $\stoptime$ on an instance $Y$, and $s_i(v_1,\dots,v_i) = 0$ if $v_i < \max\{v_1,\dots,v_{i-1}\}$. Then if $x_i = \PP(\stoptime = i)$ for all $i\in [n]$, the probability of $\stoprule$ stopping on $v_{(1)}$ having seen $v_{(2)}$ is at most $\sum_{i = 1}^n\frac{(i - 1)i}{(n - 1)n}x_i$.
\end{lemma}

\begin{proof}
First, note that $\PP(\randorder(1) = \stoptime,\randorder(2)<\stoptime) =  \sum_{i=1}^n \PP(\randorder(1) = i,\randorder(2) < i|\stoptime = i) \cdot \PP(\stoptime = i)$. Next,
\begin{align*}
     \PP(\randorder(1)=i, \randorder(2) < i| \stoptime = i) 
    = & \frac{\PP(\stoptime = i, v_i \geq v_{\max}^{<i}, \randorder(1) = i, \randorder(2) < i)}{\PP(\stoptime = i, v_i \geq v_{\max}^{<i})}\\
    = & \frac{\PP(\stoptime = i|v_i \geq v_{\max}^{<i}, \randorder(1) = i, \randorder(2) < i) \cdot \PP(v_i \geq v_{\max}^{<i},\randorder(1) = i,\randorder(2) < i)}{\PP(\stoptime = i|v_i \geq v_{\max}^{<i})\cdot \PP(v_i \geq v_{\max}^{<i})}\\
    = & \frac{\PP(v_i \geq v_{\max}^{<i},\randorder(1) = i, \randorder(2) < i)}{\PP(v_i \geq v_{\max}^{<i})} \tag{By Lemma~\ref{lem:tauequalsi}}\\
    = & \PP(\randorder(1) = i, \randorder(2) < i| v_i \geq v_{\max}^{<i}).
\end{align*}
The first equality is because stopping at $i$ when $v_i$ is not the highest so far does not increase the probability. Then we get
\begin{align*}
    \PP(\randorder(1) = \stoptime,\randorder(2) < \stoptime) 
    = & \sum_{i=1}^n\PP(\randorder(1) = i, \randorder(2) < i|v_i \geq v_{\max}^{<i})x_i \\
        = & \sum_{i=1}^n\PP(\randorder(2) < i|v_i \geq v_{\max}^{<i}, \randorder(1) = i) \cdot\PP(\randorder(1) = i|v_i \geq v_{\max}^{<i})x_i \\
    = & \sum_{i=1}^n\frac{(i-1)i}{(n-1)n}x_i.\qedhere
\end{align*}
\end{proof}
Now we can bound the objective value of the LP using duality.

\begin{restatable}{rLem}{lemLPbound}\label{lem:LPbound}
For any $n \geq 2$, the optimal value of the LP is at most $\frac{1}{4}+\frac{2}{n}$.
\end{restatable}
\begin{proof}[Proof sketch, full proof in Appendix~\ref{app:lbrobust}] By weak duality, it is sufficient to find a feasible dual solution $y$ with a dual objective
value at most $\frac{1}{4}+\frac{2}{n}$, which we do with
    \[y_i=\begin{cases}
        0 & \text{if} i < \lceil \frac{n}{2}\rceil \\
        \frac{1}{n}(1-\frac{2(n-i)}{n-1}) & \text{if}  \lceil \frac{n}{2}\rceil\leq i\leq n 
    \end{cases} \qedhere \] 
\end{proof}

Finally we put it all together to get the desired upper bound for $\PP(\randorder(1) = \stoptime, \randorder(2) < \stoptime)$ for any 0-value-oblivious stopping rule.
\begin{lemma}\label{thm:LPbound}
    For any $n \geq 2$ and any $0$-value-oblivious stopping rule with stopping time $\stoptime$ on some $Y$, then for all values $v_1,\dots,v_i \in Y$, $\PP(\randorder(1) = \stoptime, \randorder(2) < \stoptime) \leq \frac{1}{4} + \frac{2}{n}.$
\end{lemma}
\begin{proof}
    Observe that stopping on any $v_i$ if $v_i < \max\{v_1,\dots,v_{i-1}\}$ does not increase the desired probability. Then it is sufficient to consider $0$-value-oblivious stopping rules such that $s_i(v_1,\dots,v_i) =  0$ if $v_i < \max\{v_1,\dots,v_{i-1}\}$. If we let $x_i = \PP(\stoptime = i)$ for every $i\in[n]$, which we know by Lemma~\ref{lem:algtoLPfeas} to be a feasible solution to the LP, and $OPT(LP)$ represent the optimal (maximum) possible objective value we have $\PP(\randorder(1) = \stoptime, \randorder(2) < \stoptime) \leq OPT(LP) \leq \frac{1}{4} + \frac{2}{n}$
    by Lemma~\ref{lem:algtoLPobj} and Lemma~\ref{lem:LPbound}. 
\end{proof}

\subsubsection{Finishing the proof of the main lemma}\label{sec:finishlemma}

We  construct instances with large gaps between the second and third highest values.

\begin{definition}
    For an instance $Y \subseteq \mathbb{N}$, let $\infset$ be the family of instances  such that $v_{(2)} \geq n^2 v_{(3)}$.
\end{definition}

 Note that $\infset$ is infinitely large for any $n > 0$ and any infinitely large $Y$. The following lemma completes the proof of Lemma~\ref{thm:1/4}.

\begin{restatable}{rLem}{lemInV} \label{lem:InV}
    Let $\stoprule$ be a stopping rule over $n\geq 2$ values with stopping time $\stoptime$. If $\stoprule'$ is the corresponding stopping rule with stopping time $\stoptime'$ that is $\frac{1}{n^2}$-value-oblivious on infinite set $Y\subseteq \mathbb{N}$, then on any instance $V \in \infset$, $\E[v_{\max}^{<\stoptime}|v_\stoptime \geq v_{\max}^{<\stoptime}]\cdot \PP(v_\stoptime \geq v_{\max}^{<\stoptime})\leq \left(\frac{1}{4} + \frac{4}{n}\right)v_{(2)}.$
\end{restatable}

\subsubsection{Finishing the proof of the main lower bound result}
In order to apply results on stopping rules to an auction $M \in \mathcal{M}_u$, we  express the probability of stopping at any timestep as $s_i(v_1,\dots,v_i) = \PP(v_i\geq p_i(v_1,\dots,v_{i-1}) | v_j < p_j(v_1,\dots,v_{j-1}) \text{ for all }j < i).$ 
Next, we bound the revenue gained  from any bidder that does not have the highest value so far.
\begin{restatable}{rLem}{lempostbelowhighest}
\label{lem:postbelowhighest}
Given auction $M\in \mathcal{M}_u$ over $n$ agents with stopping rule $\stoprule$ and stopping time $\stoptime$, if the corresponding stopping rule $\stoprule'$ is $\frac{1}{n^2}$-value-oblivious on infinite set $Y \subseteq \mathbb{N}$, then there exists an instance $V \in \infset$ such that $\E[p_\stoptime(\textbf{v}_{:\stoptime - 1})|v_\stoptime \in [p_\stoptime(\textbf{v}_{:\stoptime - 1}),v_{\max}^{<\stoptime})]\cdot\PP(v_\stoptime \in [p_\stoptime(\textbf{v}_{:\stoptime - 1}),v_{\max}^{<\stoptime})) \leq \frac{3}{n^2}v_{(2)}$.
\end{restatable}

Now we can prove the main theorem.
\begin{proof}[Proof of Theorem~\ref{thm:atmostmsfbound}] 
 Let $\stoptime$ represent the stopping time of the stopping rule $\stoprule$ corresponding to $M_u$, which by definition of $\mathcal{M}$ means that the item is sold, then we get the following:
    \begin{align*}
        \E[\rev(M(\mathbf{v}))] =  \E[p_\stoptime(\mathbf{v}_{:\stoptime - 1})]\nonumber  
        = & \E[p_\stoptime(\mathbf{v}_{:\stoptime - 1})|v_\stoptime\geq p_\stoptime(\mathbf{v}_{:\stoptime-1})]\cdot\PP(v_\stoptime\geq p_\stoptime(\mathbf{v}_{:\stoptime-1})) \nonumber\\
        = & \E[p_\stoptime(\mathbf{v}_{:\stoptime - 1})|v_\stoptime \geq v_{\max}^{<\stoptime}]\cdot\PP(v_\stoptime \geq v_{\max}^{<\stoptime}) + \nonumber\\
        &\E[p_\stoptime(\mathbf{v}_{:\stoptime - 1})|v_\stoptime \in [p_\stoptime(\mathbf{v}_{:\stoptime - 1}),v_{\max}^{<\stoptime})]\cdot\PP(v_\stoptime \in [p_\stoptime(\mathbf{v}_{:\stoptime - 1}),v_{\max}^{<\stoptime}))\nonumber\\
        \leq & \E[v_{\max}^{<\stoptime}|v_\stoptime \geq v_{\max}^{<\stoptime}]\cdot\PP(v_\stoptime \geq v_{\max}^{<\stoptime}) + \nonumber\\
        &\E[p_\stoptime(\mathbf{v}_{:\stoptime - 1})|v_\stoptime \in [p_\stoptime(\mathbf{v}_{:\stoptime - 1}),v_{\max}^{<\stoptime})]\cdot\PP(v_\stoptime \in [p_\stoptime(\mathbf{v}_{:\stoptime - 1}),v_{\max}^{<\stoptime})).
    \end{align*}
    The inequality comes from the fact that all finite prices $\mathcal{M}_{u}$ auctions post at time $i$ cannot exceed the best value seen before $i$.  For $\stoprule'$ corresponding to $\stoprule$, let $Y \subseteq N$ be the infinite set on which it is $\frac{1}{n^2}$-value-oblivious, which we know to exist by Lemma~\ref{lem:valoblivfromorder}. Then if we consider the instance $V \in \infset$ from Lemma~\ref{lem:postbelowhighest} and apply Lemma~\ref{lem:InV} we get
\begin{align*}
    \E_{\mathbf{v} \sim \mu(V)}[\rev(M(\mathbf{v}))] &  \leq \E[v_{\max}^{<\stoptime}|v_\stoptime \geq v_{\max}^{<\stoptime}]\cdot\PP(v_\stoptime \geq v_{\max}^{<\stoptime}) \\
    & + \E[p_\stoptime(\textbf{v}_{:\stoptime - 1})|v_\stoptime \in [p_\stoptime(\textbf{v}_{:\stoptime - 1}),v_{\max}^{<\stoptime})]\cdot\PP(v_\stoptime \in [p_\stoptime(\textbf{v}_{:\stoptime - 1}),v_{\max}^{<\stoptime}))\\
    & \leq \left(\frac{1}{4} + \frac{4}{n} \right) v_{(2)} + \frac{3}{n^2} v_{(2)} \leq \left(\frac{1}{4} + \frac{7}{n}\right) v_{(2)}. \qedhere
\end{align*}

\end{proof}

\paragraph{Beyond Up to Max-Previously-Seen auctions.} The challenge in extending our $1/4$ bound to general auctions is that the infinite size set of values over which  value-obliviousness holds  can have arbitrarily large gaps between the values in that set. Thus, instead of posting a price equal to the maximum value seen so far, which intuitively is the price an auction should post when there are no time overlaps between the agents, an auction could potentially post twice that price and achieve a higher revenue for such instances with large gaps between values. We do, however, conjecture that $1/4$ is the right answer for the setting without predictions.

\subsection{Tight Impossibility Result for Consistency-Robustness Trade-off}\label{sec:lbtrade-off}
 Optimality results for secretary and online auction problems are often obtained through LP duality arguments \citep{BJS08,BJS10,AGKK23}.
The LP formulations for these problems, including ours in Section~\ref{sec:lbrobust}, rely on  history-independence properties. In the setting with predictions, these history-independence properties do not hold because, for example, the probability of a bidder accepting
a price equal to the prediction crucially depends on how many bidders have previously rejected
an offered price equal to the prediction. Such dependencies make it challenging to give an LP
formulation of our problem with predictions. Nevertheless,
 we show that the robustness-consistency trade-off achieved by our auction is optimal for the following family of auctions.
 


\begin{definition}
\label{def:mechfamily}
Consider the following allocation rules for bidder $i$:
\begin{itemize}
    \item $a_1^i$ never allocates the item to bidder $i$,
    \item for all $j \in [i-1]$, $a_{2, j}^i$ allocates to $i$ if $v_i \geq \max\{\prediction, v^{<i}_{(j)}\}$, and
    \item for all $j \in [i-1]$, $a_{3, j}^i$ allocates to $i$ if $v_i \geq v^{<i}_{(j)}$.
\end{itemize}
Let $A^i = \{a_1^i\} \cup \{a_{2,j}^i\}_{j \in [i-1]} \cup \{a_{3,j}^i\}_{j \in [i-1]}$.  An auction $M$ is in the family of  Prediction or Any-Previously-Seen (PA) auctions $\mathcal{M}_{a}$ if, for every bidder $i \in [n]$, there is an allocation rule $a^i \in A^i$ such that, if the item is not allocated to a bidder $j < i$, $M$ allocates to $i$ according to $a^i$.
\end{definition}

We note that our~\mech\ Auction, as well as the auctions from~\citet{BJS10} and~\citet{HKP04}, are in this family of auctions (but not our error-tolerant auction). 
 The main result in this section is the following.

\begin{restatable}{rThm}{thmHardness}\label{thm:hardness}
For any $\alpha \in [0,1]$, there is no  auction $M$  in the PA family of auctions $\mathcal{M}_{a}$ that is $\alpha$-consistent and   $(\frac{1-\alpha^2}{4} + \omega(\frac{1}{n}))$-robust.
\end{restatable}


\paragraph{Overview of the proof.} We say that an $\alpha$-consistent auction $M$ is robustness-optimal among $\mathcal{M}$ if $M \in \mathcal{M}$ and there is no $\alpha$-consistent auction $M' \in \mathcal{M}$ that achieves strictly better robustness than $M$. The full proof can be found in Appendix~\ref{app:lbtrade-off}.
\begin{enumerate}
\item We first show that, for any PA auction $M \in \mathcal{M}_{a}$, there exists an auction from a simpler family of auctions that achieves consistency and robustness that are no worse than those achieved by $M$ (Section~\ref{sec:reduction}). Thus, impossibility results for this simpler family extend to PA auctions.
\item We then show that, for any $\alpha \in [0,1]$, there exist $\first, \second \in [n]$ and an $\alpha$-consistent, robustness-optimal auction among this simpler family of auctions that posts price $\infty$ at each time $i \in [1, \first]$, then price $\max\{\prediction, v^{\leq\first}_{(1)}\}$ at each time $i \in [\first + 1, \second]$, and finally price $v^{\leq\second}_{(1)}$ at each time $i \in [\second + 1, n]$ (Section~\ref{sec:mainlemma}). This is the main part of the proof. 
\item Finally, we show that, for the auction structure described in the second step, the optimal thresholds for maximizing robustness are $\first = \frac{1 -  \alpha}{2}n$ and $\second = \frac{\alpha + 1}{2}n$, achieving robustness at most $\frac{1 - \alpha^2}{4} + \omega(\frac{1}{n})$
(Section~\ref{sec:optthresholds}). 
\end{enumerate}



\bibliographystyle{abbrvnat}
\bibliography{biblio}

\appendix
\section{Missing proofs from Section~\ref{sec:strategyproof} (Strategyproofness)}\label{app:strategyproof}

\obssinglethreshold*
\begin{proof}
If $\tau = \bestsofar{\first}$ in Line~\ref{line:first update}, it means that $\bestsofar{\first} \geq \prediction$. First note that $\tau$ will never equal to $\prediction$, implying $p = \tau$. Consider the case where the winner is above the threshold before the $\second$-th bidder departs, then the allocation is terminated before the possible update of $\tau$, making $p = \tau = \bestsofar{\first}$. Consider the case where the winner, 
 denoted as $i^*$ is above the threshold after the departure of the $\second$-th bidder, then $\bestsofar{\second} = \bestsofar{\first}$, since there are no bidder with value higher than $\tau = \bestsofar{\first}$ before bidder $i^*$, in this case the payment is also $\bestsofar{\first}$.
\end{proof}

\lemvalbidderfirst*
\begin{proof}
We first consider the case where $v_i < \tau$. Note that if such a bidder reports her type truthfully, she won't receive the item. The only way she can possibly win the item is by reporting a value $\hat{v}_i \geq \tau$. However, since, as observed in Lemma~\ref{obs:singlethreshold}, in the case of a single threshold, the price $p = \tau \geq v_i$, bidder $i$ would incur non-positive utility if she were to obtain the item.

Now consider the case where $v_i > \tau$. By Lemma~\ref{lem:firstphase}, we already have that the bidders with the first $\first$ departures has no incentive to misreport. Note that the threshold $\tau$, and therefore the payment $p$, is independent of the rest of the bidders' reports. First if bidder $i$, with her true type, wins the auction, then she has no incentive to misreport since the price is independent of her report. We now consider some bidder $i$ with $v_i \geq \tau$ but bidder $i$ is not the winner. This can only result from one of the following two scenarios: 1. Some other bidder $j$ is above the threshold before bidder $i$ is above the threshold and 2. bidder $i$ and some bidder $j$ are above the threshold at the same time, but bidder $i$ loses in tie-breaking.
We will now address these two cases separately.

\textbf{Case one:} Since in the single-threshold case $\tau$ is not updated after Line~\ref{line:first update}, it means that bidder $j$ is above the threshold before the arrival time of bidder $i$. For bidder $i$ to claim the item, she would need to report an earlier arrival time $\hat{a}_i < a_i$ to be above the threshold weakly earlier than bidder $j$. However, this is not a feasible outcome since we assume bidders cannot report $\hat{a}_i < a_i$.

\textbf{Case two:} First, note that by the definition of our auction and the fact that $v_i \geq \tau$,
the only possible tie-breaking happens at Line~\ref{line:tiebreaking} after $\tau$ changes from $\infty$ to $\bestsofar{\first}$. Since the tie-breaking rules are independent of bidders' reports, there is no way for bidder $i$ to change the results of the tie-breaking. Additionally, since the threshold is $\infty$ before this point, there is also no way for bidder $i$ to misreport her type to win the item.
\end{proof}

\lemvalbiddertwo*

\begin{proof}
We first consider the case where $v_i \leq \tau_2$. Note that if such a bidder reports her type truthfully, she won't receive the item. The only way she can possibly win the item is by reporting a value $\hat{v}_i \geq \tau_2$. However, since the price $p \geq \tau_2 > v_i$, bidder $i$ would incur negative utility if she were to obtain the item.

We now consider the case where $v_i \geq \tau_1$. First consider such a bidder that is not the winner. First note that if $i \leq \first$ by Lemma~\ref{lem:firstphase} she has no incentive to misreport. Consider any $i \geq \first$ but $i$ is not the winner. Then it can  only result from one of the following two scenarios: 1. Some other bidder $j$ is above the threshold before bidder $i$ is above the threshold. 2. bidder $i$ and some bidder $j$ are above the threshold at the same time, but bidder $i$ loses in tie-breaking.
We will now address these two cases separately.

\textbf{Case one:} Since $v_i \geq \tau_1$, it means that bidder $j$ is above the threshold before the arrival time of bidder $i$. For bidder $i$ to claim the item, she would need to report an earlier arrival time $\hat{a}_i < a_i$ to be above the threshold weakly earlier than bidder $j$. However, this is not a feasible outcome since we assume bidders cannot report $\hat{a}_i < a_i$.

\textbf{Case two:} First, note that by the definition of our auction and the fact that $v \geq \tau_1$, the only possible tie-breaking happens at Line~\ref{line:tiebreaking} after $\tau$ changes from $\infty$ to $\bestsofar{\first}$. Since the tie-breaking rules are independent of bidders' reports, there is no way for bidder $i$ to change the results of the tie-breaking. Additionally, since the threshold is $\infty$ before this point, there is also no way for bidder $i$ to misreport her arrival as earlier to win the item.
\end{proof}

\section{Missing Proofs from Section~\ref{sec:error} (Error Tolerant Auction)}\label{app:error}
\thmError*
\begin{proof}
We first prove that the worst-case expected revenue is always bounded by $\frac{1-\alpha^2}{4}v_{(2)}$. The argument is identical to the proof of Lemma~\ref{lem:robustness}, but changing the case conditions from $\prediction > \vmax$ and $\prediction \leq \vmax$ to $\gamma\prediction > \vmax$ and $\gamma\prediction \leq \vmax$, respectively. 

We now prove the $\alpha\gamma\quality v_{(1)}$ bound when the prediction is relatively accurate, i.e., $\quality \geq \gamma$. First by the definition of $\quality$ we have $\quality \prediction \leq \vmax \leq \frac{\prediction}{\quality}$. Suppose $\quality \geq \gamma$, then $\gamma \prediction \leq \quality \prediction \leq \vmax$. Consider the same instance we consider in Lemma~\ref{lem:consistency}, where the highest value agent $i^* \in [\first+1, \second]$, i.e., agent $i^*$ in the second phase. First note that she will get the item and pay $\tau$, since $\vmax \geq \gamma \prediction$ and $\vmax \geq \vbest$. In addition, the threshold $\tau = \max(\gamma\prediction, \vbest) \geq \gamma\prediction$. In this case the revenue achieved is at least 
$$\tau \geq  \gamma\prediction \geq \gamma\quality v_{(1)} .$$
Since such instance happens with probability $\frac{\second - \first}{n} = \alpha$. We therefore have that the expected revenue is at least $\alpha\gamma\quality \vmax$ when $\quality \geq \gamma$.

Lastly, we note that the strategyproof argument for auction \emech\ is identical to that for auction \mech.
\end{proof}

\section{Missing proofs from Section~\ref{sec:lbrobust} (Impossibility without Predictions)}\label{app:lbrobust}

\lemorderoblivious*
We now prove Lemma~\ref{lem:orderoblivious} by constructing a stopping rule $\stoprule'$ from any given stopping rule $\stoprule$, using a construction from \citet{correa}, and show that $\stoprule'$ from such construction satisfies the conditions above. Let $V_i = \{v_1, \ldots, v_i\}$ be the unordered set of the first $i$ values.

\begin{definition}
\label{def:rprime}
    Given an arbitrary stopping rule $\stoprule$ with stopping time $\stoptime$, define the stopping rule $\stoprule'$ to be such that, for all time $i \in [n]$ and  values $w_1,\dots,w_i$,
  \[s_i'(w_1,\dots,w_i) = \PP(\stoptime = i| V_{i-1} = \{w_1, \ldots, w_{i-1}\}, v_i = w_i, \stoptime \geq i).\]

\end{definition}
Intuitively, $\stoprule'$ is a weighted average of $\stoprule$ over permutations of $w_1,\dots,w_{i-1}$. From this section onwards, any use of $\stoprule'$ and $\stoptime'$ refers to this construction, corresponding to whatever stopping rule $\stoprule$ is used in that context.

The following lemma is proven in~\citet{correa}.

\begin{lemma}[Equation (6) inside proof of Lemma 2 in \citet{correa}] 
\label{lem:correa}
For any  step $i \in [n]$, $i-1$ unordered values $W_{i-1} \subseteq \mathbb{R}_{\geq 0}$, value $w_i\geq 0$, and stopping rule $\stoprule$ with stopping time $\stoptime$ with corresponding stopping rule $\stoprule'$ with stopping time $\stoptime'$ we have that
\[\PP(\stoptime \geq i | V_{i-1} = W_{i-1}, v_i = w_i) = \PP(\stoptime' \geq i | V_{i-1} = W_{i-1}, v_i = w_i)\]
\end{lemma}
We are now ready to prove Lemma~\ref{lem:orderoblivious}.
\begin{proof}[Proof of Lemma~\ref{lem:orderoblivious}] Let $\stoprule'$ be the stopping rule from Definition~\ref{def:rprime}. For 1., note that   for any $\randorder \in \mathcal{S}_{i-1}$, 
\begin{align*}s'_i(w_1,\dots,w_i) & = \PP(\stoptime = i|V_{i-1} = \{w_1,\dots,w_{i-1}\},v_i = w_i,\stoptime \geq i) \\
 & = \PP(\stoptime = i|V_{i-1} =\{w_{\randorder(1)},\dots,w_{\randorder(i-1)}\},v_i = w_i,\stoptime \geq i)\\
 & = s_i'(w_{\randorder(1)},\dots,w_{\randorder(i-1)},w_i). 
\end{align*}
Thus, $\stoprule'$ is order-oblivious. For 2., consider a stopping rule $\stoprule$ with stopping time $\stoptime$, an instance of  $n$ values $W\subseteq \mathbb{R}_{\geq 0}$, and $i\in[n]$ such that $\PP(\stoptime = i) = 0$. This implies $s_i'(w_1,\dots,w_i) = 0$ for every $\{w_1,\dots,w_i\} \subseteq  W$. Thus, $\PP(\stoptime' = i) = 0$.  

For 3., consider the stopping time $\stoptime$ of stopping rule $\stoprule$ over $n$ values $W$, we have that
    \begin{align*}
        \E[v_{\max}^{<\stoptime}|v_{\stoptime}\geq v_{\max}^{<\stoptime}]\cdot\PP(v_{\stoptime}\geq v_{\max}^{<\stoptime}) 
        = & \sum_{i=1}^n\E[v_{\max}^{<i}|v_i\geq v_{\max}^{<i},\stoptime = i]\cdot\PP(v_i\geq v_{\max}^{<i}, \stoptime = i)\\
        = & \sum_{i=1}^n \sum_{S \subseteq W : |S| = i}\E[v_{\max}^{<i}|v_i\geq v_{\max}^{<i}, \stoptime = i, V_i = S ] \\
        & \qquad \qquad \qquad \qquad \cdot \PP(v_i\geq v_{\max}^{<i}, \stoptime = i | V_i = S)\cdot\PP(V_i = S)  
         \end{align*}
    For a fixed set $T \subseteq W $ such that $|T| = i$, we let $\{t_1, \ldots, t_i\} = T$ such that $t_1 <  \ldots<  t_i$. Then,
    \begin{align*}
         & \sum_{T \subseteq W : |T| = i}\E[v_{\max}^{<i}|v_i\geq v_{\max}^{<i}, \stoptime = i, V_i = T ]\cdot\PP(v_i\geq v_{\max}^{<i}, \stoptime = i | V_i = T)\cdot\PP(V_i = T)    \\
         = & \sum_{T \subseteq W : |T| = i}t_{i-1} \cdot \PP(\stoptime = i | V_i = T,  v_i = t_i)\cdot\PP(V_i = T| v_i = t_i)\cdot\PP(v_i = t_i) 
    \end{align*}
Next, note that 
\begin{align*}
    \PP(\stoptime = i |  V_i = T,  v_i = t_i)
    = &  \PP(\stoptime = i | V_i = T,  v_i = t_i, \stoptime \geq i)\cdot\PP(\stoptime \geq i | V_i = T,  v_i = t_i) \\ 
    = & \PP(\stoptime' = i | V_i = T,  v_i = t_i, \stoptime \geq i)\cdot\PP(\stoptime' \geq i | V_i = T,  v_i = t_i) \\
    = & \PP(\stoptime' = i |  V_i = T,  v_i = t_i) \numberthis
\end{align*}
where the second equality is by definition of $\stoptime'$ and  Lemma~\ref{lem:correa}. We conclude that
\begin{align*}
         \E[v_{\max}^{<\stoptime}|v_{\stoptime}\geq v_{\max}^{<\stoptime}]\cdot\PP(v_{\stoptime}\geq v_{\max}^{<\stoptime})  
         = & \sum_{i=1}^n\sum_{T \subseteq W : |T| = i}t_{i-1} \cdot \PP(\stoptime = i | V_i = T,  v_i = t_i) \\
         & \qquad \qquad \qquad \qquad \cdot\PP(V_i = T| v_i = t_i)\cdot\PP(v_i = t_i) \\
          = & \sum_{i=1}^n \sum_{T \subseteq W : |T| = i}t_{i-1} \cdot \PP(\stoptime' = i | V_i = T,  v_i = t_i) \\
          & \qquad \qquad \qquad \qquad \cdot\PP(V_i = T| v_i = t_i)\cdot\PP(v_i = t_i) \\
           = & \E[v_{\max}^{<\stoptime'}|v_\stoptime'\geq v_{\max}^{<\stoptime}]\cdot\PP(v_\stoptime'\geq v_{\max}^{<\stoptime}),
        \end{align*}
where the last equality follows identically as for $\stoptime$.
\end{proof}

\lemalgtoLPfeas*
\begin{proof}
     First observe that $x_i = \PP(\stoptime = i) = \PP(\stoptime = i, v_i \geq v_{\max}^{<i})$ because we assume $\PP(\stoptime = i|v_i < \max\{v_1,\dots,v_{i-1}\}) = 0$. Then for each $i\in[n]$ we can break this value down as follows:
     \begin{align}
         \PP(\stoptime = i, v_i \geq v_{\max}^{<i}) = &\PP(\stoptime = i, v_i \geq v_{\max}^{<i} | \stoprule\text{ reaches i})\cdot \PP(\stoprule\text{ reaches i}) \nonumber \\
        \leq & \PP(\stoptime = i, v_i \geq v_{\max}^{<i}|\stoprule\text{ reaches i})\cdot (1-\sum_{j<i}x_j) \nonumber\\
        = & \PP(\stoptime = i|v_i \geq v_{\max}^{<i} ,\stoprule\text{ reaches i}) \cdot \PP(v_i \geq v_{\max}^{<i}|\stoprule \text{ reaches i}) \cdot (1-\sum_{j < i}x_j) \nonumber\\
        = & s_i\cdot \PP(v_i \geq v_{\max}^{<i}) \cdot (1-\sum_{j<i}x_j) \nonumber\\
        = & \frac{s_i}{i} (1-\sum_{j<i}x_j) \label{eq:nonneg} \\
        \leq &\frac{1}{i}(1-\sum_{j<i}x_j) \label{eq:constraint}.
     \end{align}
Observe that in the fourth line, the relative order of $v_i$ is independent of the ordering of $v_1,\dots,v_{i - 1}$. Clearly (\ref{eq:nonneg}) is nonnegative because events $\PP(\stoptime = i)$ for distinct $i$'s are disjoint, and (\ref{eq:constraint}) is precisely the first line of constraints.
\end{proof}

\lemtauequalsi*
\begin{proof}
    Observation : if some event $A'$ is independent of the relative order of $v_1,\dots,v_i$, then 
    \begin{equation} \label{eq:observation1}
        \PP(\stoptime = i|\stoptime \geq i, v_i \geq v_{\max}^{<i}, A') = s_i 
    \end{equation}
   because $\{\stoptime = i\}$ only depends on $v_1,\dots,v_i$, specifically only whether $v_i \geq v_{\max}^{<i}$ by definition of $s_i(v_1,\dots,v_i)$. Let us define $A$ to be the event that $\{\randorder(1) = i,\randorder(2) < i\}$. Observe that this is independent of the relative order of $v_1,\dots,v_j$ for any $j < i$. 
   
   We show by induction that $$\PP(\stoptime \geq j|v_i \geq v_{\max}^{<i},A) = \PP(\stoptime \geq j|v_i \geq v_{\max}^{<i})$$ for any $j \leq i$. The base case is  $\PP(\stoptime \geq 1|v_i \geq v_{\max}^{<i},A) = s_1 = \PP(\stoptime \geq 1|v_i \geq v_{\max}^{<i})$    because $s_1(v) = s_1$ for any $v$. Now assume for any $j\leq i$ that the statement holds for $j-1$. We first perform the following decomposition by applying the observation. 
    \begin{align*}
        \PP(\stoptime = j-1|\stoptime \geq j-1,v_i \geq v_{\max}^{<i}, A) = &\PP(\stoptime = j-1| \stoptime \geq j-1, v_{j-1} \geq v_{\max}^{<j-1},v_i \geq v_{\max}^{<i}, A) \\
        & \qquad \qquad \qquad \qquad \cdot \PP(v_{j-1} \geq v_{\max}^{<j-1}|\stoptime \geq j-1, v_i \geq v_{\max}^{<i}, A) \\
        = & s_{j-1}\PP(v_{j-1} \geq v_{\max}^{<j-1})  \\
        = &  \PP(\stoptime = j-1|\stoptime \geq j-1, v_{j-1} \geq v_{\max}^{<j-1}, v_i \geq v_{\max}^{<i}) \cdot \PP(v_{j-1} \geq v_{\max}^{<j-1})
    \end{align*}
where the first equality is since  $\PP(\stoptime = j-1 | v_j < v_{\max}^{<j-1}\}) = 0$. Then we can get the following.
    \begin{align*} 
    \PP(\stoptime \geq j|v_i \geq v_{\max}^{<i}, A)
    = & \PP(\stoptime \neq j-1|\stoptime \geq j-1, v_i \geq v_{\max}^{<i}, A)\cdot \PP(\stoptime \geq j-1|v_i \geq v_{\max}^{<i}, A)\\ 
    = & (1- \PP(\stoptime = j-1|\stoptime \geq j-1,v_i \geq v_{\max}^{<i}, A)) \cdot \PP(\stoptime \geq j-1|v_i \geq v_{\max}^{<i}, A) \\ 
    = & (1 - \PP(\stoptime = j-1| \stoptime \geq j-1, v_{j-1} \geq v_{\max}^{<j-1},v_i \geq v_{\max}^{<i}) \cdot \PP(v_{j-1} \geq v_{\max}^{<j-1})) \\
    & \qquad \qquad \qquad \qquad \qquad \cdot \PP(\stoptime \geq j-1|v_i \geq v_{\max}^{<i}, A) 
    \end{align*}
Now we can apply the induction hypothesis and see that this is equivalent to the following.
    \begin{align*}
    & (1 - \PP(\stoptime = j-1|\stoptime \geq j-1, v_{j-1} \geq v_{\max}^{<j-1}, v_i \geq v_{\max}^{<i}) \cdot \PP(v_{j-1} \geq v_{\max}^{<j-1})) \\ 
    & \qquad \qquad \qquad \qquad \qquad \cdot\PP(\stoptime \geq j-1|v_i \geq v_{\max}^{<i})
    \end{align*}
We can apply in reverse the same computation as for $\PP(\stoptime \geq j | v_i \geq v_{\max}^{<i}, A)$ to see that this equals $\PP(\stoptime \geq j|v_i \geq v_{\max}^{<i})$ as desired. Then by letting $j = i$ we have that
\begin{equation}
    \PP(\stoptime \geq i|v_i \geq v_{\max}^{<i}, A) = \PP(\stoptime \geq i|v_i \geq v_{\max}^{<i}).
\end{equation}
 Then we have that 
    \begin{align*}
        \PP(\stoptime = i|v_i \geq v_{\max}^{<i},\randorder(1) = i, \randorder(2) < i) 
         = & \PP(\stoptime = i|\stoptime\geq i,v_i \geq v_{\max}^{<i},\randorder(1) = i, \randorder(2) < i) \\
         & \qquad \qquad \qquad \cdot \PP(\stoptime \geq i|v_i \geq v_{\max}^{<i},\randorder(1) = i, \randorder(2) < i) \\
        = & s_i\cdot \PP(\stoptime \geq i|v_i \geq v_{\max}^{<i},\randorder(1) = i, \randorder(2) < i) \\
        = & s_i \cdot \PP(\stoptime \geq i|v_i \geq v_{\max}^{<i}) \\
        = & \PP(\stoptime = i|\stoptime \geq i, v_i \geq v_{\max}^{<i})\cdot\PP(\stoptime \geq i|v_i \geq v_{\max}^{<i})  \\
        = & \PP(\stoptime = i|v_i \geq v_{\max}^{<i}).\qedhere
    \end{align*}
\end{proof}
\lemLPbound*
\begin{proof}
    By weak duality, it is sufficient to find a feasible dual solution $y$ with a dual objective value at most $\frac{1}{4} + \frac{2}{n}$. Consider the following:
    \[y_i=\begin{cases}
        0 & i < \lceil \frac{n}{2}\rceil \\
        \frac{1}{n}(1-\frac{2(n-i)}{n-1}) & \lceil \frac{n}{2}\rceil\leq i\leq n
    \end{cases}.\]

    First we show this is feasible. Clearly $y$ is nonnegative, so let us focus on \[y_i \geq \frac{i-1}{(n-1)n}-\frac{1}{i}\sum_{j=i+1}^ny_j.\]

   Fix $i \geq \lceil \frac{n}{2} \rceil$, then
    \begin{align*}
        \sum_{j=i+1}^ny_j = & \sum_{j=i+1}^n\frac{1}{n}\left(1-\frac{2(n-j)}{n-1}\right) \\
        = & \frac{(n-i)(n-1) - 2\sum_{j=i+1}^n(n-j)}{n(n-1)} \\
        = & \frac{(n-i)(n-1)-(n-i)(n-i-1)}{n(n-1)} = \frac{i(n-i)}{n(n-1)}.
    \end{align*}

    In the constraint this gives us 
    \[\frac{i-1}{(n-1)n}-\frac{1}{i}\sum_{j=i+1}^ny_j = \frac{i-1}{(n-1)n}-\frac{i(n-i)}{in(n-1)} =\frac{2i-n-1}{n(n-1)} = \frac{1}{n}(1-\frac{2(n-i)}{n-1})=y_i. \]

    Now for $i < \lceil\frac{n}{2}\rceil$, observe that we want to show $0\geq\frac{i-1}{n(n-1)}-\frac{1}{i}\sum_{j=\lceil\frac{n}{2}\rceil}^ny_j$. The following is sufficient.
    \[\sum_{j=\lceil{\frac{n}{2}\rceil}}^ny_j = \frac{(\lceil\frac{n}{2}\rceil - 1)(n-\lceil\frac{n}{2}\rceil+1)}{n(n-1)}= \frac{(\lceil\frac{n}{2}\rceil - 1)\lfloor\frac{n}{2}\rfloor + 1)}{n(n-1)}\geq \frac{i(i-1)}{n(n-1)}\]
    The inequality is because $i <\lceil\frac{n}{2}\rceil \leq \lfloor \frac{n}{2}\rfloor + 1$. Thus $q$ is a feasible solution to the dual problem.

    The objective value is as follows:
    \begin{align*}
        \sum_{i=1}^ny_i = & \sum_{i=\lceil\frac{n}{2}\rceil}^n\frac{1}{n}\left(1-\frac{2(n-i)}{n-1}\right) \\
        = & \frac{n-\lceil\frac{n}{2}\rceil + 1}{n}-\frac{2}{n(n-1)}\sum_{i=1}^{n-\lceil\frac{n}{2}\rceil}i\\
        = & 1 - \frac{\lceil\frac{n}{2}\rceil - 1}{n}-\frac{(n-\lceil\frac{n}{2}\rceil)(n-\lceil\frac{n}{2}\rceil + 1)}{n(n-1)}\\
        \leq & \frac{1}{2} +\frac{1}{n} - \frac{(\frac{n}{2}-1)\frac{n}{2}}{n(n-1)} = \frac{1}{4}+ \frac{1}{n} + \frac{1}{4n-4} \leq \frac{1}{4} + \frac{2}{n}.\qedhere
    \end{align*}
\end{proof}
\lemInV*
\begin{proof}
We know that $\stoprule'$ and the corresponding infinite set $Y$ exist by Lemma~\ref{lem:valoblivfromorder}.
\begin{align*}
&\E[v_{\max}^{<\stoptime}|v_{\stoptime}\geq v_{\max}^{<\stoptime}]\cdot\PP(v_{\stoptime}\geq v_{\max}^{<\stoptime})  = \E[v_{\max}^{<\stoptime'}|v_{\stoptime'}\geq v_{\max}^{<\stoptime'}]\cdot\PP(v_{\stoptime'}\geq v_{\max}^{<\stoptime'}) \tag{by Lemma~\ref{lem:orderoblivious}}\\
    = & v_{(2)} \cdot \PP(\randorder(1) = \stoptime',\randorder(2) < \stoptime') + \E[v_{\max}^{<\stoptime'}|v_{\stoptime'}\geq v_{\max}^{<\stoptime'}, \randorder(1) \neq \stoptime' \cup \randorder(2) \geq \stoptime']\\
    & \qquad \qquad \qquad \qquad \qquad \cdot \PP(v_{\stoptime'}\geq v_{\max}^{<\stoptime'}, \randorder(1) \neq \stoptime' \cup \randorder(2) \geq \stoptime')\\
    \leq & v_{(2)} \cdot \PP(\randorder(1) = \stoptime', \randorder(2)<\stoptime') + v_{(3)} \\
    \leq & v_{(2)}\cdot \left(\PP(\randorder(1) = \stoptime', \randorder(2) < \stoptime') + \frac{1}{n^2}\right) 
\end{align*}
The first inequality is because $\randorder(2) \geq \stoptime'$ implies $v_{\max}^{<\stoptime'} \leq v_{(3)}$, and $\randorder(2) < \stoptime'\cap\randorder(1) \neq \stoptime'$ implies that $v_{\stoptime'} \leq v_{(3)}$. The last line is because $v_{(2)} \geq n^2v_{(3)}$ by definition of $\infset$.  Since $\stoprule'$ is $\varepsilon$-value-oblivious for $\varepsilon = \frac{1}{n^2}$, we know that $s_i'(v_1,\dots,v_i) \in[y_i-\varepsilon,y_i+\varepsilon)$ for all $i\in[n]$ for some $y_i\in[0,1]$. We define a $0$-value-oblivious stopping rule $\hat{\stoprule}$ with stopping rule $\hat{\stoptime}$ as follows: $\hat{s}_i(v_1,\dots,v_i) = s_i$ if $v_i > \max\{v_1,\dots,v_{i - 1}\}$ and 0 otherwise (because stopping on $v_i$ if it is not the highest so far does not increase the probability of stopping on $v_{(1)}$). The union bound argument in \cite{correa} gives us
\begin{equation*} \PP(\randorder(1) = \stoptime',\randorder(2) < \stoptime') \leq \PP(\randorder(1) = \hat{\stoptime}, \randorder(2) < \hat{\stoptime}) + n\varepsilon = \PP(\randorder(1) = \hat{\stoptime}, \randorder(2) < \hat{\stoptime}) + \frac{1}{n}.\end{equation*}
We conclude that
\begin{align*}
\E[v_{\max}^{<\stoptime} |v_\stoptime \geq v_{\max}^{<\stoptime}]\cdot\PP(v_\stoptime \geq v_{\max}^{<\stoptime}) 
& \leq v_{(2)}\left(\PP(\randorder(1) = \stoptime', \randorder(2) < \stoptime') + \frac{1}{n^2}\right)  \\
& \leq v_{(2)}\left(\PP(\randorder(1) = \hat{\stoptime}, \randorder(2) < \hat{\stoptime}) + \frac{1}{n} + \frac{1}{n^2}\right) \\
& < v_{(2)}\left(\frac{1}{4} + \frac{2}{n} + \frac{1}{n} + \frac{1}{n^2}\right) \\
& \leq v_{(2)}\left(\frac{1}{4} + \frac{4}{n}\right)
\end{align*}
where the first two inequalities are by the above two series of inequalities and the third inequality is by Lemma~\ref{thm:LPbound}.
\end{proof}

\lempostbelowhighest*
We will first need a helper lemma for the specific case of posting to the second highest bidder having seen the highest already.

\begin{lemma}\label{lem:seenhighest}
    Given auction $M \in \mathcal{M}_u$ over $n$ agents with stopping rule $\stoprule$ and stopping time $\stoptime$, if the corresponding stopping rule $\stoprule'$ is $\frac{1}{n^2}$-value-oblivious on infinite set $Y \subseteq \mathbb{N}$, then there exists an instance $V \in \infset$ such that 
    \[\E[p_\stoptime(\textbf{v}_{:\stoptime-1})|v_{(2)} \geq p_\stoptime(\textbf{v}_{:\stoptime - 1}), v_\stoptime = v_{(2)}, v_{\max}^{<\stoptime} = v_{(1)}]\cdot\PP(v_{(2)} \geq p_\stoptime(\textbf{v}_{:\stoptime - 1})|v_\stoptime = v_{(2)}, v_{\max}^{<\stoptime} = v_{(1)}) \leq \frac{2}{n^2}v_{(2)}.\]
\end{lemma}
\begin{proof}
    Let $l \in \mathbb{N}$ such that $l \geq n^2$. Since $Y$ is infinitely large and over $\mathbb{N}$, there exist  values 
\[\{v_{(1)}, v_{(2)}^1, \dots, v_{(2)}^l, v_{(3)}, \dots,v_{(n)}\} \in Y\] such that $v_{(1)} > v_{(2)}^1$, $ v_{(2)}^{j} \geq n^2 \cdot v_{(2)}^{j+1}$ for  $j\in[l-1]$, $v_{(2)}^l \geq n^2 \cdot v_{(3)} $, and $v_{(j)} \geq v_{(j+1)}$ for $j\in \{3, \ldots, n-1\}$. We define instances $V^j = \{v_{(1)}, v_{(2)}^j, v_{(3)}, \dots, v_{(n)}\}$ for every $j\in[l]$, and note that $V^j \in \infset$ by construction. Next, we define  intervals $\priceinterval_j = (v_{(2)}^{j+1}, v_{(2)}^j]$ for $j\in[l-1]$, $\priceinterval_l = (v_{(3)}, v_{(2)}^l]$, and $\priceinterval_{l+1} = [0, v_{(3)}]$. Fix an instance $j\in [l]$, and if we let $A_\stoptime$ be the event that $v_\stoptime = v_{(2)}$ and $v_{\max}^{<\stoptime} = v_{(1)}$, we get
\begin{align*}
    & \E[p_\stoptime(\textbf{v}_{:\stoptime - 1})|v_{(2)}^j \geq p_\stoptime(\textbf{v}_{:\stoptime - 1}), A_\stoptime] \cdot \PP(v_{(2)}^j \geq p_\stoptime(\textbf{v}_{:\stoptime - 1})|A_\stoptime) \\
    = & \sum_{k=1}^{l+1} \E[p_\stoptime(\textbf{v}_{:\stoptime-1})|v_{(2)}^j\geq p_\stoptime(\textbf{v}_{:\stoptime - 1}), p_\stoptime(\textbf{v}_{:\stoptime - 1}) \in \priceinterval_k, A_\stoptime] \\
    & \qquad \qquad \qquad \qquad\cdot \PP(v_{(2)}^j\geq p_\stoptime(\textbf{v}_{:\stoptime - 1})| p_\stoptime(\textbf{v}_{:\stoptime - 1}) \in \priceinterval_k, A_\stoptime) \cdot  \PP(p_\stoptime(\textbf{v}_{:\stoptime - 1}) \in \priceinterval_k|A_\stoptime)
\end{align*}
Note that $\PP(v_{(2)}^j\geq p_\stoptime(\textbf{v}_{:\stoptime - 1})|p_\stoptime(\textbf{v}_{:\stoptime - 1}) > v_{(2)}^1)=0$ for all $j\in[l]$. Now if $p_\stoptime(\textbf{v}_{:\stoptime - 1}) \in \priceinterval_k$ for $k\neq j$, we consider two cases. If $k > j$, then $p_\stoptime(\textbf{v}_{:\stoptime - 1}) \leq v_{(2)}^k \leq \frac{1}{n^2} v_{(2)}^j$. If $k < j$, then $p_\stoptime(\textbf{v}_{:\stoptime - 1}) > v_{(2)}^{k+1} \geq v_{(2)}^j$. Thus $$\E[p_\stoptime(\textbf{v}_{:\stoptime - 1})|v^j_{(2)} \geq p_\stoptime(\textbf{v}_{:\stoptime - 1}), p_\stoptime(\textbf{v}_{:\stoptime - 1})\in \priceinterval_k, A_\stoptime]\cdot\PP(v^j_{(2)} \geq p_\stoptime(\textbf{v}_{:\stoptime - 1})|p_\stoptime(\textbf{v}_{:\stoptime - 1})\in \priceinterval_k, A_\stoptime) \leq \frac{1}{n^2}v_{(2)}^j,$$ for all $k\neq j$. Now, since $\sum_{k=1}^{l}\PP(p_\stoptime(\textbf{v}_{:\stoptime - 1}) \in \priceinterval_k|A_\stoptime) \leq 1$,
 there must exist some $j'\in[l]$ such that $\PP(p_\stoptime(\textbf{v}_{:\stoptime-1})\in \priceinterval_{j'}|A_\stoptime) \leq \frac{1}{l}$. Then if we fix instance $V^{j'}$ we get
\begin{align*}
    & \sum_{k=1}^{l+1} \E[p_\stoptime(\textbf{v}_{:\stoptime-1})|v_{(2)}^{j'}\geq p_\stoptime(\textbf{v}_{:\stoptime - 1}), p_\stoptime(\textbf{v}_{:\stoptime - 1}) \in \priceinterval_k, A_\stoptime]\\
    & \qquad \qquad \qquad \qquad \cdot\PP(v_{(2)}^{j'}\geq p_\stoptime(\textbf{v}_{:\stoptime - 1})| p_\stoptime(\textbf{v}_{:\stoptime - 1}) \in \priceinterval_k,A_\stoptime) \cdot \PP(p_\stoptime(\textbf{v}_{:\stoptime - 1}) \in \priceinterval_k|A_\stoptime)\\
    = & \sum_{k\in[l+1]: k\neq j'}\E[p_\stoptime(\textbf{v}_{:\stoptime-1})|v_{(2)}^{j'}\geq p_\stoptime(\textbf{v}_{:\stoptime - 1}), p_\stoptime(\textbf{v}_{:\stoptime - 1}) \in \priceinterval_k, A_\stoptime]\\
    &\qquad \qquad \qquad \qquad \cdot\PP(v_{(2)}^{j'}\geq p_\stoptime(\textbf{v}_{:\stoptime - 1}), p_\stoptime(\textbf{v}_{:\stoptime - 1}) \in \priceinterval_k,A_\stoptime)  \cdot \PP(p_\stoptime(\textbf{v}_{:\stoptime - 1}) \in \priceinterval_k|A_\stoptime) \\
    & \qquad \qquad + \E[p_\stoptime(\textbf{v}_{:\stoptime -1})|v_{(2)}^{j'} \geq p_\stoptime(\textbf{v}_{:\stoptime - 1}), p_\stoptime(\textbf{v}_{:\stoptime - 1}) \in \priceinterval_{j'}, A_\stoptime]\\
    & \qquad \qquad \qquad \qquad \cdot\PP(v_{(2)}^{j'} \geq p_\stoptime(\textbf{v}_{:\stoptime - 1})), p_\stoptime(\textbf{v}_{:\stoptime - 1})\in \priceinterval_{j'}, A_\stoptime) \cdot \PP(p_\stoptime(\textbf{v}_{:\stoptime - 1})\in \priceinterval_{j'}| A_\stoptime) \\
    \leq & \frac{1}{n^2}v_{(2)}^{j'} + \frac{1}{l}v_{(2)}^{j'} \leq \frac{2}{n^2} v_{(2)}^{j'}. \qedhere
\end{align*}
\end{proof}
We can now finish the proof of Lemma~\ref{lem:postbelowhighest}
\begin{proof}[Proof of Lemma~\ref{lem:postbelowhighest}]
We first break down the value
\[\E[p_\stoptime(\mathbf{v}_{:\stoptime - 1})|v_\stoptime \in [p_\stoptime(\mathbf{v}_{:\stoptime - 1}),v_{\max}^{<\stoptime})]\cdot\PP(v_\stoptime \in [p_\stoptime(\mathbf{v}_{:\stoptime - 1}),v_{\max}^{<\stoptime}))\]
into two cases.

Case 1: $v_\stoptime \leq v_{(3)}$ or $v_{\max}^{<\stoptime} \leq v_{(3)}$. Then $p_\stoptime(\mathbf{v}_{:\stoptime - 1})$ is bounded above by $v_{(3)}$.

Case 2: $v_\stoptime, v_{\max}^{<\stoptime} > v_{(3)}$. Observe that this is only possible if $v_\stoptime = v_{(2)}$ and $v_{\max}^{<\stoptime} = v_{(1)}$. Then conditioning on this, we get 
\begin{align*}
    &\E[p_\stoptime(\mathbf{v}_{:\stoptime - 1}) | v_{\stoptime} \in [p_\stoptime(\mathbf{v}_{:\stoptime - 1}), v_{\max}^{<\stoptime}), v_{\stoptime} = v_{(2)}, v_{\max}^{<\stoptime} = v_{(1)}]\cdot\PP(v_{\stoptime} \in [p_{\stoptime}(\mathbf{v}_{:\stoptime - 1}), v_{\max}^{<\stoptime})|v_\stoptime = v_{(2)}, v_{\max}^{<\stoptime} = v_{(1)}) \\
    = & \E[p_\stoptime(\mathbf{v}_{:\stoptime - 1})|v_{(2)} \geq p_\stoptime(\mathbf{v}_{:\stoptime - 1}), v_\stoptime = v_{(2)}, v_{\max}^{<\stoptime} = v_{(1)}] \cdot \PP(v_{(2)} \geq p_\stoptime(\mathbf{v}_{:\stoptime - 1})|v_\stoptime = v_{(2)}, v_{\max}^{<\stoptime} = v_{(1)}).
\end{align*}
If we consider instance $V$ from Lemma~\ref{lem:seenhighest}, this value is at most $\frac{2}{n^2}v_{(2)}$. Also note that $V \in \infset$, so $v_{(3)} \leq \frac{1}{n^2}v_{(2)}$.

Thus we get 
\begin{align*}
    & \E[p_\stoptime(\mathbf{v}_{:\stoptime - 1})|v_\stoptime \in [p_\stoptime(\mathbf{v}_{:\stoptime - 1}),v_{\max}^{<\stoptime})]\cdot\PP(v_\stoptime \in [p_\stoptime(\mathbf{v}_{:\stoptime - 1}),v_{\max}^{<\stoptime})) \\
    \leq & v_{(3)} + \E[p_\stoptime(\textbf{v}_{:\stoptime-1})|v_{(2)} \geq p_\stoptime(\textbf{v}_{:\stoptime - 1}), v_\stoptime = v_{(2)}, v_{\max}^{<\stoptime} = v_{(1)}] \\
    & \qquad \qquad \qquad \qquad \qquad \cdot\PP(v_{(2)} \geq p_\stoptime(\textbf{v}_{:\stoptime - 1})|v_\stoptime = v_{(2)}, v_{\max}^{<\stoptime} = v_{(1)}) \\
    \leq & \frac{3}{n^2}v_{(2)}. \qedhere
\end{align*}
\end{proof}

\section{Full proof of Theorem~\ref{thm:hardness}}\label{app:lbtrade-off}
\thmHardness*

\subsection{The Reduction from PA to PM Auctions}
\label{sec:reduction}

We start by giving a simple formula for the consistency and robustness of auctions in $\mathcal{M}_{a}$. For $M \in \mathcal{M}_a$, we let $C^M :=\{ \randorder: \text{$M$ posts price $\prediction$ to $\randorder(1)$ under order $\randorder$}\}$ and $R^M := \{ \randorder : \text{$M$ posts price $v_{(2)}$ to $\randorder(1)$ under order $\randorder$} \}$, where $\mu(V)$ represents a random matching of the values to indices in $[n]$.

\begin{lemma}
\label{lem:CR}
An  auction $M \in \mathcal{M}_{a}$ achieves consistency $|C^M| / (n!)$ and robustness $|R^M| / (n!)$.
\end{lemma}
\begin{proof}

Observe that for each $\randorder \in C^M$, the auction achieves revenue $v_{(1)}$ when the prediction is correct. In addition, for each $\randorder \in R^M$, the auction achieves revenue $v_{(2)}$, even when the prediction is incorrect. Then consistency and robustness are lower bounded by the probability of drawing $\randorder \in C^M$ and $\randorder \in R^M$ respectively, which are precisely $|C^M| / (n!)$ and $|R^M| / (n!)$.

To show that consistency is at most $|C^M| / (n!)$, it suffices to find a single instance where equality holds. Consider the instance where the values are $v_{(1)} = 1, v_{(2)} = \cdots = v_{(n)} = 0$ and prediction $\prediction = 1$ (we will denote this $I_1$). Observe from our definition of PA auctions, the only (noninfinite) prices that can be posted to bidder $i$ are in the set $\{v_1, v_2, \ldots,  v_{i - 1}, \prediction\}$. For this instance, only two of these may be nonzero, $\prediction$ or $v_j = 1$ if the highest bidder arrives at step $j < i$ bidder. In the first case, the only bidder who can accept this price is the highest bidder, and revenue of $v_{(1)}$ is extracted. In the second case, we know the highest bidder has already departed, so bidder $i$ must have value $v_i = 0$ and no revenue can be gained. Thus the only way revenue is gained in this instance is by posting $\prediction$ to the highest bidder at step $\randorder(1)$, and the revenue is precisely $v_{(1)}$, so consistency is exactly $|C^M| / (n!)$.

Similarly, for robustness consider the instance where the values are $v_{(1)} = 1, v_{(2)} = \varepsilon, v_{(3)} = \cdots = v_{(n)} = 0$ for some $\varepsilon < 1$ and the prediction is $\prediction = v_{(1)} + 1$ (we will denote this $I_2$). No revenue is gained by posting $\prediction$ since no bidder would accept that price. The only other positive prices that can be posted to bidder $i$ are $v_j = 1$ if the highest bidder arrives at step $j < i$ or $v_j = \varepsilon$ if the second highest bidder arrives at step $j < i$.  The first case is the same as above. As for $v_j = \varepsilon$, the only bidder who can accept this price is the highest bidder, and revenue of $v_{(2)}$ is gained. Thus the only way revenue is gained in this instance is by posting price $v_{(2)}$ to the highest bidder at step $\randorder(1)$, and the revenue is precisely $v_{(2)}$, so robustness is $|R^M| / (n!)$.
\end{proof}

We now reduce to the following simpler family of mechanisms.
\begin{definition}
Consider the following three allocation rules: 
 $a_1^i$  never allocates the item to bidder $i$, $a_2^i$ allocates to $i$ if $v_i \geq \max\{\prediction, \vmaxless{i}\}$, and $a_3^i$ allocates to $i$ if $v_i \geq \vmaxless{i}$.  An auction $M$ is in the family of  Prediction or Max-Previously-Seen (PM) auctions $\mathcal{M}_{m}$ if, for every bidder $i \in [n]$, there is an allocation rule $a^i \in \{a_1^i, a_2^i, a_3^i\}$ such that if the item is not allocated to a bidder $j < i$ then $M$ allocates to $i$ according to $a^i$.
\end{definition}

Observe that $\mathcal{M}_{m} \subset \mathcal{M}_{a} \subset \mathcal{M}$.

\begin{lemma}
\label{lem:PAtoPM}
For every $M \in \mathcal{M}_{a}$, there exists some $M' \in \mathcal{M}_{m}$ such that $\text{consistency}(M') \geq \text{consistency}(M) \text{ and robustness}(M') \geq \text{robustness}(M).$
\end{lemma}
\begin{proof}
We will construct $M'$ from $M$ as follows. We determine the allocation rule $M'$ uses for bidder $i$:
\[a'^i = \begin{cases}
    a_1^i & a^i \in \{a_1^i\} \\
    a_2^i & a^i \in \{a_{2,j}^i\}_{j\in [i - 1]} \\
    a_3^i & a^i \in \{a_{3,j}^i\}_{j\in [i - 1]}
\end{cases}\]
Note that $a_2^i$ and $a_3^i$ allocate the item to bidder i if $v_i$ is at least $\max\{\prediction, v_{(1)}^{<i}\}$ and $v_{(1)}^{<i}$ respectively.

Next, we show that $\text{consistency}(M') \geq \text{consistency}(M)$. First, we have that $\text{consistency}(M') = |C^{M'}| / (n!)$ by  Lemma~\ref{lem:CR}, so it is sufficient to show that $|C^{M'}| \geq |C^M|$. Consider any $\randorder \in C^M$ with $\randorder(1) = i$. Observe that if $M$ does not allocate the item prior to step $i$, neither does $M'$ because at any $j < i$, $M'$ posts to $j$ a price at least as high as the price $M$ posts to $j$. Since price $\prediction$ is posted by $M$ to $i$, we know that $a^i \in \{a_{2,j}^i\}_{j\in [i - 1]}$, and subsequently $a'^i = a_2^i$, so $M'$ also posts price $\prediction$ to $i$. Thus $\randorder \in C^{M'}$, and therefore $|C^{M'}| \geq |C^M|$.

Similarly, we show that $\text{robustness}(M') \geq \text{robustness}(M)$ by proving that $|R^{M'}| \geq |R^M|$. Consider any $\randorder \in R^M$ with $\randorder(1) = i$. By the same argument as above, if $M$ does not allocate the item prior to step $i$ neither does $M'$. Since $M$ posts price $v_{(2)}$ to $i$, there are two cases for $a^i$. Case 1 is $a^i = a_{2,1}^i$ if $\prediction \leq v_{(2)}$. Note that we know $j = 1$ because $v_{(1)}$ must be seen at time i. Then $a'^i = a_2^i$ and also posts $v_{(2)}$ to bidder $i$. Case 2 is $a^i = a_{3,1}^i$, and $j = 1$ by the same reasoning. Then $a'^i = a_3^i$ and again posts $v_{(2)}$ to bidder $i$. Thus $\randorder \in R^{M'}$ and $|R^{M'}| \geq |R^M|$.
\end{proof}

By Lemma~\ref{lem:PAtoPM}, impossibility results for $\mathcal{M}_{m}$ extend to $\mathcal{M}_{a}$.

\subsection{The Main Lemma for the Impossibility Result}
\label{sec:mainlemma}

The main lemma for the impossibility result shows that there exists an $\alpha$-consistent auction that is robustness-optimal among  auctions in $\mathcal{M}_{a}$ and has a three-phase structure (as does our auction).

\begin{lemma} \label{lem:optmech}
    There exists an $\alpha$-consistent auction that is robustness-optimal among  auctions in $\mathcal{M}_{a}$ and satisfies the following structure: it posts price $\infty$ at each time $i \in [1, \first]$, then price $\max\{\prediction, \vmaxless{i}\}$ at each time $i \in [\first + 1, \second]$, and finally price $\vmaxless{i}$ at each time $i \in [\second + 1, n]$.
\end{lemma}

The remainder of Section~\ref{sec:mainlemma} is devoted to the proof of Lemma~\ref{lem:optmech}. 

\paragraph{Overview of the proof of Lemma~\ref{lem:optmech}.}
The proof follows an interchange argument that shows that if an auction $M \in \mathcal{M}_m$ does not post prices in the order specified by Lemma~\ref{lem:optmech}, then there are two positions $i$ and $i+1$ that violate this order and the prices posted at these time steps can be swapped without decreasing $|C^M|$ and $|R^M|$, and therefore without decreasing consistency and robustness.  There are three potential violations of the ordering specified by Lemma~\ref{lem:optmech}. In Lemma~\ref{lem:swap2and3}, we consider the case where $\vmaxless{i}$ is posted to bidder $i$ and $\max\{\prediction, \vmaxless{i + 1}\}$ to bidder $i+1$, in Lemma~\ref{lem:swap1and2} the case where $\vmaxless{i}$ is posted to bidder $i$ and $\infty$ to bidder $i+1$, and in Lemma~\ref{lem:swap1and3} the case where $\max\{\prediction, \vmaxless{i}\}$ to bidder $i$ and $\infty$ to bidder $i+1$. 

We now define the interchange function $f_i: \mathcal{S}_n\to \mathcal{S}_n$. For fixed index $i$ and any permutation $\randorder$, for every $j \in [n]$ let 
\[f_i(\randorder)(j) = \begin{cases}
    i+1 & \randorder(j) = i \\
    i & \randorder(j) = i + 1 \\
    \randorder(j) & \text{else},
\end{cases}\]
which is a bijective function that swaps the values of the ith and (i+1)th bidders. We first state a trivial fact regarding the revenue achieved from the first $i-1$ bidders for two auctions that are identical before step $i$. This fact will be repeatedly used in the proof of the next lemmas.

\begin{lemma} \label{beforei}
    Consider two  auctions $M, M' \in \mathcal{M}_m$ that are identical for steps before $i$. Then $M$ under order $\randorder$ and $M'$ under order $f_i(\randorder)$ gain the same revenue before step $i$.
\end{lemma}

\begin{proof}
    Observe that $f_i(\randorder)$ does not affect the values that appear before $i$. Then $M'$ sees the same ranks before $i$ under $f_i(\randorder)$ as $M$ does under $\randorder$, and since they follow the same rules the revenue gained at each step before $i$ is the same. 
\end{proof}

The first potential violation of the ordering specified by Lemma~\ref{lem:optmech} is when $\vmaxless{i}$ is posted to bidder $i$ and $\max\{\prediction, \vmaxless{i + 1}\}$ to bidder $i+1$.

\begin{lemma} \label{lem:swap2and3}
 Consider an  auction $M \in \mathcal{M}_m$ that posts price  $\vmaxless{i}$ at some step $i$ and $\max\{\prediction, \vmaxless{i + 1}\}$ at step $i + 1$. Let $M'$ be the same auction as $M$ except that it posts price $\max\{\prediction, \vmaxless{i}\}$  at step $i$ and $\vmaxless{i + 1}$  at step $i + 1$. Then, $|C^{M'}| \geq |C^M|$ and $|R^{M'}| \geq |R^M|$.
\end{lemma}
\begin{proof}    First we show that if $\randorder \in C^M$, then $f_i(\randorder) \in  C^{M'}$, meaning $f_i(C^M) \subseteq C^{M'}$. Since $f_i$ is a bijective function, consequently $|C^{M'}| \geq |C^M|$. Assume that $\prediction = v_{(1)}$.  Consider any $\randorder \in C^M$, then $\prediction$ is posted to the highest ranked bidder who arrives at step $\randorder(1)$. Observe that to prove $f_i(\randorder) \in C^{M'}$, it is sufficient to show that price $\prediction$ is posted to bidder $f_i(\randorder)(1)$, or that revenue $\prediction$ is extracted by $M'$ under ordering $f_i(\randorder)$. There are three cases.

    First, if $\randorder(1)<i$,  observe that there is no difference  between $M$ and $M'$  before step $i$. Then it follows from Lemma \ref{beforei} that since $M$ extracts revenue $\prediction$ at step $\randorder(1) < i$, $M'$ extracts the same revenue under $f_i(\randorder)$ before step $i$. 

    The second case is if $\randorder(1)\in \{i,i+1\}$. Recall that $\randorder \in C^M$ implies that price $\prediction$ is posted to bidder $\randorder(1)$. Then since $M$ posts $\vmaxless{i} < \prediction$ to bidder $i$, $\randorder(1)=i+1$ is the only possibility, and indeed $M$ posts $\max\{\prediction, \vmaxless{i+1}\} = \prediction$ to $i + 1$. Now we must show that $M'$ posts price $\prediction$ to $f_i(\randorder)(1)$. By our definition of the interchange function, $f_i(\randorder)^{-1}(i) = \randorder^{-1}(i + 1) = 1$, so $f_i(\randorder)(1) = i$. Note that if $M$ reaches step $i$, then $M'$ does as well. Then since $M'$ posts price $\max\{\prediction, \vmaxless{i}\} = \prediction$ at step $i$, then $f_i(\randorder) \in C^{M'}$.

    The third and last case is if $\randorder(1) >i+1$. Observe that it is sufficient to show that under $M'$, bidders $i$ and $i+1$ do not receive the item because at steps after $i + 1$, the auctions see the same order of bidders and make the same posts. Clearly if under $M$ bidder $i + 1$ does not accept its posted-price $\prediction$, ie $v_{i + 1} < \prediction$, then under $M'$, bidder $i$ with value $v_{(f_i(\randorder)^{-1}(i))} = v_{i + 1}$ will not accept its posted price $\prediction$. Now we consider bidder $i + 1$ under $M'$. If we let $\vmaxlessbar{\ell}$ be the value of the highest ranked bidder seen before step $\ell$ given ordering $f_i(\randorder)$, we can see that $\vmaxlessbar{i+1} \geq \vmaxlessbar{i} = \vmaxless{i}$. Then if bidder $i$ does not accept the price $\vmaxless{i}$ posted under $M$, ie $v_i < \vmaxless{i}$, then $v_{(f_i(\randorder)^{-1}(i + 1))}=v_{i} < \vmaxless{i} \leq \vmaxlessbar{i+1}$ and under $M'$ bidder $i + 1$ also does not accept its posted price $\vmaxlessbar{i + 1}$. 
 
     Next, to show the second part of lemma, we show that $f_i(R^M) \subseteq R^{M'}$. If $\randorder \in R^M$ is true, then $\vmaxless{i}=v_{(2)}$ is posted to the bidder at step $\randorder(1)$. Similar to the consistency proof, to prove $f_i(\randorder) \in R^{M'}$, it is sufficient to show that price $v_{(2)}$ is posted to bidder $f_i(\randorder)(1)$, or that revenue $v_{(2)}$ is extracted by $M'$ under ordering $f_i(\randorder)$ (for any value of $\prediction$). Cases 1 and 3 are exactly the same as for consistency. Then consider $\randorder(1) \in \{i,i+1\}$.

 If $\prediction > v_{(1)}$, then under $M$ price $\max\{\prediction, \vmaxless{i+1}\} > v_{(1)}$ is posted to and rejected by bidder $i + 1$, so $\randorder(1) = i$ is the only possibility. Indeed, price $\vmaxless{i} = v_{(2)}$ may be posted to $i$ if the second highest bidder arrives before $i$. After the interchange, bidder $i$ is offered price $\prediction > v_{(1)}$ under $M'$ and they reject it. Observe that if $\randorder(1) = i$, then $f_i(\randorder)(1) = i + 1$. Then under $M'$, bidder $i + 1$ sees price $\vmaxlessbar{i + 1} = \vmaxless{i} = v_{(2)}$.

 Now if $\prediction < v_{(2)}$, then there are two scenarios. First, consider $\randorder(1) = i$, then since $M$ posts $\vmaxless{i}$ at $i$ we know that $\vmaxless{i} = v_{(2)}$. Observe that since the highest two bidders arrive by step $i$, then $v_{i + 1} < v_{(2)}$. Under $f_i(\randorder)$, bidder $i$ has value $v_{(f_i(\randorder)^{-1}(i))} = v_{i + 1}$. Then when $M'$ posts to bidder $i$ price $\max\{\prediction, \vmaxlessbar{i}\} = \vmaxless{i} = v_{(2)}$, it is rejected. $M'$ then posts price $\vmaxlessbar{i + 1}$, which is exactly $v_{(2)}$ because bidder $i$ has value below $v_{(2)}$, to the bidder with rank $f_i(\randorder)^{-1}(i + 1) = \randorder^{-1}(i) = 1$. If instead $\randorder(1) = i + 1$, then it is impossible for bidder $i$ to have the second highest value or else they would accept their price $\vmaxless{i} < v_{(2)}$. Then for $\max\{\prediction, \vmaxless{i}\} = v_{(2)}$ to hold, the second highest bidder must arrive before $i$ and $\vmaxless{i} = v_{(2)}$. Thus under $M'$, price $\vmaxlessbar{i} = \vmaxless{i}$ is posted at step $i$ to the bidder with rank $f_i(\randorder)^{-1}(i) = \randorder^{-1}(i + 1) = 1$.

 Observe that if $v_{(2)} \leq \prediction \leq v_{(1)}$, given that $\vmaxless{i} = v_{(2)}$, selling by posting $\max\{\prediction, \vmaxless{i}\}$ and $\vmaxless{i}$ both result in at least $v_{(2)}$ revenue, so swapping these two prices does not lower robustness.
\end{proof}

The second potential violation of the ordering specified by Lemma~\ref{lem:optmech} is when  $\vmaxless{i}$ is posted to bidder $i$ and $\infty$ to bidder $i+1$.

\begin{lemma} \label{lem:swap1and2}
 Consider an  auction $M \in \mathcal{M}_m$ that posts price  $\vmaxless{i}$ at some step $i$ and $\infty$ at step $i + 1$. Let $M'$ be the same auction at $M$ except that it posts price $\infty$  at step $i$ and $\vmaxless{i + 1}$  at step $i + 1$. Then, $|C^{M'}| \geq |C^M|$ and $|R^{M'}| \geq |R^M|$.
\end{lemma}
\begin{proof} We first show that $|C^{M'}| \geq |C^M|$ by proving that $f_i(C^M) \subseteq C^{M'}$. Let $\prediction = v_{(1)}$. Consider any $\randorder \in C^M$, so $\prediction$ is posted to bidder $\randorder(1)$. There are again three cases. The first case, $\randorder(1) < i$, is as in Lemma~\ref{lem:swap2and3}. The second case is $\randorder(1)\in \{i,i+1\}$. Observe that it is impossible to post $\prediction$ to bidder $\randorder(1)$ by posting $\infty$ or $\vmaxless{i}$. This is because $\vmaxless{i} \leq v_{(2)}$. The third and last case is $\randorder(1) > i + 1$. It is sufficient to show that $M'$ does not sell the item at time $i$ or $i + 1$. Clearly the former is true because $\infty$ is posted to $i$. Since $M$ fails to sell the item at step $i$ by posting price $\vmaxless{i}$, we know that $v_i < \vmaxless{i}$. Then under $M'$, the price $\vmaxlessbar{i + 1} \geq \vmaxlessbar{i} = \vmaxless{i}$ is posted to bidder with value $v_{(f_i(\randorder)^{-1}(i + 1))} = v_i$, so the item is not sold to bidder $i + 1$ under $M'$.

    For the second part of the lemma, we show $f_i(\randorder)\in R^{M'}$ for any $\randorder \in R^M$. If $\randorder \in R^M$ is true, then $\vmaxless{i}=v_{(2)}$ is posted to the bidder at step $\randorder(1)$. Similar to the consistency proof, to prove $f_i(\randorder) \in R^{M'}$, it is sufficient to show that price $v_{(2)}$ is posted to bidder $f_i(\randorder)(1)$, or that revenue $v_{(2)}$ is extracted by $M'$ under ordering $f_i(\randorder)$ (for any value of $\prediction$).  Once again, cases one and three are the same. Now consider $\randorder(1)\in\{i,i+1\}$. Since bidder $i + 1$ never accepts $\infty$, this means that $\randorder(1)=i$. $M$ posts price $\vmaxless{i}$ to bidder $i$, so $\vmaxless{i} = v_{(2)}$. Observe that $f_i(\randorder)^{-1}(i + 1) = \randorder^{-1}(i) = 1$. First $M'$ posts $\infty$ at step $i$, and then $\vmaxlessbar{i + 1}$, which equals $v_{(2)}$ since the highest bidder cannot be at step $i$ under $f_i(\randorder)$. Then $v_{(2)}$ is posted to bidder $i + 1  = f_i(\randorder)(1)$.
\end{proof}

The third and last potential violation of the ordering specified by Lemma~\ref{lem:optmech} is when  $\max\{\prediction, \vmaxless{i}\}$ to bidder $i$ and $\infty$ to bidder $i+1$.

\begin{lemma} \label{lem:swap1and3}
 Consider an  auction $M \in \mathcal{M}_m$ that posts price  $\max\{\prediction, \vmaxless{i}\}$ at some step $i$ and $\infty$ at step $i + 1$. Let $M'$ be the same auction at $M$ except that it posts price $\infty$  at step $i$ and $\max\{\prediction, \vmaxless{i + 1}\}$  at step $i + 1$. Then, $|C^{M'}| \geq |C^M|$ and $|R^{M'}| \geq |R^M|$.
\end{lemma}
\begin{proof}  We again first show that $|C^{M'}| \geq |C^M|$ by proving that $f_i(C^M) \subseteq C^{M'}$. Let $\prediction = v_{(1)}$. Consider any $\randorder \in C^M$, so $\prediction$ is posted to bidder $\randorder(1)$. The first case, $\randorder(1) < i$, is as in Lemma~\ref{lem:swap2and3}. The second case is $\randorder(1) \in \{i, i+1\}$. Since auction $M$ posts $\infty$ to $i+1$, $\randorder(1) = i$ is the only possibility, and indeed $\max\{\prediction, \vmaxless{i}\} = \prediction$ is posted at $i$. First observe that if $M$ reaches step $i$, then $M'$ does as well. $M'$ subsequently posts $\infty$ to bidder $i$, effectively skipping them. We know that $f_i(\randorder)^{-1}(i + 1) = \randorder^{-1}(i) = 1$. Then when under $f_i(\randorder)$ auction $M'$ posts $\max\{\prediction, \vmaxless{i + 1}\} = \prediction$ at step $i + 1$, it is to the highest bidder. The third and last case is $\randorder(1) > i+1$. Use the same argument as in Lemma \ref{lem:swap1and2} except instead of $\vmaxless{i}$ being rejected by bidder $i$ under $M$ and $i + 1$ under $M'$, here $\max\{\prediction, \vmaxless{i}\} \geq \vmaxless{i}$ is being rejected by bidder $i$ under $M$ and $i + 1$ under $M'$.

For the second part of the lemma, we show that $f_i(\randorder)\in R^{M'}$ for $\randorder \in R^M$. If $\randorder \in R^M$ is true, then $\vmaxless{i}=v_{(2)}$ is posted to the bidder at step $\randorder(1)$. Similar to the consistency proof, to prove $f_i(\randorder) \in R^{M'}$, it is sufficient to show that price $v_{(2)}$ is posted to bidder $f_i(\randorder)(1)$, or that revenue $v_{(2)}$ is extracted by $M'$ under ordering $f_i(\randorder)$ (for any value of $\prediction$). Observe that the proofs for cases 1 and 3 are the same as above. Case two is impossible if $\prediction > v_{(1)}$ because no bidder accepts prices $\infty$ or $\max\{\prediction, \vmaxless{i}\} > v_{(1)}$. Then if $\prediction < v_{(1)}$, we have $\randorder(1) = i$, as auction $M$ posts $\infty$ to bidder $i + 1$. In order for value at least $v_{(2)}$ (but below $v_{(1)}$) to be posted at step $i$ by $M$, we need $\max\{\prediction, \vmaxless{i}\} \geq v_{(2)}$, so it is sufficient for the second highest bidder to arrive before $i$ and therefore $\vmaxless{i} = v_{(2)}$. Observe that $f_i(\randorder)^{-1}(i + 1) = \randorder^{-1}(i) = 1$. Under $M'$, $\infty$ is posted to bidder $i$, and then $\max\{\prediction, \vmaxlessbar{i + 1}\}$ is posted to bidder $i + 1$. We know that $\vmaxlessbar{i + 1} = \vmaxlessbar{i} = \vmaxless{i}$ by the same reasoning as in Lemma~\ref{lem:swap1and2}. Then $M'$ posts $\max\{\prediction, \vmaxless{i + 1}\} \in [v_{(1)}, v_{(2)}]$ to bidder $i + 1$ with the highest value.
\end{proof}

We are now ready to prove Lemma~\ref{lem:optmech}. \newline

\textit{Proof of Lemma~\ref{lem:optmech}.} Consider an $\alpha$-consistent robustness-optimal  auction $M \in \mathcal{M}_m$. If it does not satisfy the structure specified by Lemma~\ref{lem:optmech}, then there exist time steps $i$ and $i+1$ such that $M$ either  posts prices $\vmaxless{i}$  and $\max\{\prediction, \vmaxless{i + 1}\}$, or  prices $\infty$ and $\vmaxless{i + 1}$, or prices $\infty$ and $\max\{\prediction, \vmaxless{i + 1}\}$ to $i$ and $i + 1$. Lemma~\ref{lem:swap2and3}, Lemma~\ref{lem:swap1and2}, and Lemma~\ref{lem:swap1and3} show that for each of these cases, the two prices can be swapped without decreasing $|R^M|$ and $|C^M|$. By repeating this swapping process, we obtain an auction $M'$ such that $|C^{M'}| \geq |C^M|$ and $|R^{M'}| \geq |R^M|$. 

Thus, by Lemma~\ref{lem:CR}, $M'$ is also an $\alpha$-consistent robustness-optimal auction among $\mathcal{M}_m$ with the structure presented in Lemma~\ref{lem:optmech}. By Lemma~\ref{lem:PAtoPM}, $M'$ is also robustness-optimal among $\mathcal{M}_a$. \qed

\subsection{The Optimal Thresholds} \label{sec:optthresholds}

For  auctions in $\mathcal{M}_a$ constructed as in Lemma~\ref{lem:optmech}, the time thresholds $\first$ and $\second$ set to $\frac{1 - \alpha}{2}n$ and $\frac{1 + \alpha}{2}n$ achieve $\alpha$-consistency and $\frac{1 - \alpha^2}{4} + O(\frac{1}{n})$-robustness. We show that for $\alpha \in [0,1]$, no other thresholds lead to a better robustness, which then shows the impossibility result for  PA auctions. We note that our auction also use these same thresholds. \newline

\textit{Proof of Theorem~\ref{thm:hardness}.} Let us first introduce some notation. Let $R_i(M)$ be the event that step $i$ is reached under auction $M$ and let $P_i^{\prediction}(M)$ be the event that  $\prediction$ is posted at step $i$ under auction $M$. For a fixed $\alpha\in [0,1]$, consider an $\alpha$-consistent   auction $M \in \mathcal{M}_a$ that is optimal with respect to robustness and is structured as in Lemma~\ref{lem:optmech} with time thresholds $i_1$ and $i_2$. Observe that if $\prediction = v_{(1)}$, then the consistency achieved by $M$ is 
    \begin{align*}
        &\sum_{i=1}^n \PP(\randorder^{-1}(i) = 1) \cdot \PP(R_i(M)|\randorder^{-1}(i) = 1) \cdot 
 \PP(P_i^{\prediction}(M)|R_i(M), \randorder^{-1}(i) = 1) \\
 = &\frac{1}{n} \sum_{i=\first+1}^{\second}\PP(R_i(M)|\randorder^{-1}(i) = 1) = \frac{\second-\first}{n}.
    \end{align*}
   where the first equality is because $\prediction$ is posted only at steps $i \in [\first+1, \second]$ and the highest ranking bidder is equally likely to be at any step. The second equality holds because $M$ posting only $\infty$ before $\first$ and bidders within $[\first + 1, i - 1]$ failing to accept price $v_{(1)}$. Thus, $M$ achieves $\alpha$-consistency if $\second \geq \first + \alpha n$ .

    In our auction, we use $\first = \frac{1 - \alpha}{2}n$ and $\second = \frac{1 + \alpha}{2}n$, and we  show no other pair $\first', \second'$ can improve robustness. Recall from Lemma~\ref{lem:CR} that the robustness of $M$ is precisely $|R^M| / (n!)$. We consider two cases. The first is if $\first' \leq \first$. Observe that if $\prediction < v_{(2)}$, and letting $\randorder\sim \mathcal{S}_n$, then we have that 
    \begin{align*}
        \frac{|R^M|}{n!} = \underset{\randorder\sim \mathcal{S}_n}{\PP}(\randorder \in R^M) & = \underset{\randorder\sim \mathcal{S}_n}{\PP}(M \text{ posts price } v_{(2)} \text{ to } \randorder(1)) \\
        & = \underset{\randorder\sim \mathcal{S}_n}{\PP}(\randorder(2) \leq \first') \cdot  \underset{\randorder\sim \mathcal{S}_n}{\PP}(\randorder(1) > \first' | \randorder(2) \leq \first') \\
        & = \frac{\first'}{n} \frac{n - \first'}{n - 1}
    \end{align*}
    where the first equality is by definition of $\randorder$ and the second by definition of $R^M$. The third equality is since we need $\randorder(1) > \first'$ to not post $\infty$ to $\randorder(1)$, $\randorder(2) \leq \randorder(1)$ so that  $\max\{\prediction, \vmaxless{\randorder(1)}\} = \vmaxless{\randorder(1)} = v_{(2)}$ is posted to $\randorder(1)$, and $\randorder(2) \not \in [\first', \randorder(1)]$ to not sell to $\randorder(2)$ and reach $\randorder(1)$.
    Differentiating $\frac{\first'}{n} \frac{n - \first'}{n - 1}$ with respect to $\first'$, we get $\frac{n - 2\first'}{n(n - 1)}$, which is positive for $\first'\leq \frac{n}{2}$. Then since $\first' \leq \frac{1 - \alpha}{2}n \leq \frac{n}{2}$, we get that the robustness of $M$ is $\frac{\first'}{n}\frac{n - \first'}{n - 1} \leq \frac{\first}{n}\frac{n - \first}{n - 1} = \frac{n}{n - 1} \frac{1 - \alpha^2}{4}$.

    The second case is if $\first' > \first$. Since $\second = \first + \alpha n$, then $\second' \geq \first' + \alpha n \geq \second$. Observe that if $\prediction > v_{(1)}$, and letting $\randorder\sim \mathcal{S}_n$, then we have that
        \begin{align*}
        \frac{|R^M|}{n!}  & = \underset{\randorder\sim \mathcal{S}_n}{\PP}(M \text{ posts price } v_{(2)} \text{ to } \randorder(1)) \\
        & = \underset{\randorder\sim \mathcal{S}_n}{\PP}(\randorder(2) \leq \second') \cdot  \underset{\randorder\sim \mathcal{S}_n}{\PP}(\randorder(1) > \second' | \randorder(2) \leq \first') \\
        & = \frac{\second'}{n}\frac{n - \second'}{n-1}
    \end{align*}
    where the second equality is since we need $\randorder(1) > \second'$ to not post $\infty$ or $\prediction$ to $\randorder(1)$, $\randorder(2) \leq \randorder(1)$ so that  $\vmaxless{\randorder(1)} = v_{(2)}$ is posted to $\randorder(1)$, and $\randorder(2) \not \in [\second', \randorder(1)]$ to not sell to $\randorder(2)$ and reach $\randorder(1)$. Differentiating $\frac{\second'}{n}\frac{n - \second'}{n-1}$ with respect to $\second'$, we get $\frac{n - 2\second'}{n(n - 1)}$, which is negative for $\second' \geq \frac{n}{2}$. Since $\second' \geq \frac{1 + \alpha}{2}n \geq \frac{n}{2}$, we obtained that the robustness achieved by $M$ is $\frac{\second'}{n}\frac{n - \second'}{n - 1}\leq \frac{\second}{n}\frac{n - \second}{n - 1} = \frac{n}{n - 1} \frac{1 - \alpha^2}{4}$.

Thus, an $\alpha$-consistent   auction $M \in \mathcal{M}_a$ that is optimal with respect to robustness and is structured as in Lemma~\ref{lem:optmech} achieves a robustness that is at most $\frac{n}{n - 1} \frac{1 - \alpha^2}{4} = \frac{1 - \alpha^2}{4} + O(\frac{1}{n})$. By Lemma~\ref{lem:optmech}, we conclude that this robustness bound holds for any PA auction. \qed

\end{document}